%% file: paper.tex
\documentclass[sigplan,screen,10pt]{acmart}
\pdfoutput=1
\def\ifextended#1#2{#1}

\setcopyright{acmlicensed}
\copyrightyear{2020}
\acmYear{2020}
\acmDOI{10.1145/3423211.3425689}

\acmConference[Middleware '20]{21st International Middleware Conference}{December 7--11, 2020}{Delft, Netherlands}
\acmBooktitle{21st International Middleware Conference (Middleware '20), December 7--11, 2020, Delft, Netherlands}
\acmPrice{15.00}
\acmISBN{978-1-4503-8153-6/20/12}

\usepackage[utf8]{inputenc}
\usepackage[english]{babel}
\usepackage{xspace}
\usepackage[caption=false]{subfig}
\usepackage{tikz}
\usetikzlibrary{patterns}
\usepackage{hyphenat}

\usepackage{amsthm}

\newcommand{\system}[0]{\textsc{Spider}\xspace}
\newcommand{\channel}[0]{IRMC\xspace}
\newcommand{\channela}[0]{\mbox{\channel-RC}\xspace}
\newcommand{\channelb}[0]{\mbox{\channel-SC}\xspace}
\newcommand{\msg}[2]{\ensuremath{\langle \textsc{#1}, #2 \rangle}\xspace}
\newcommand{\smsg}[3]{\ensuremath{\langle \textsc{#1}, #2 \rangle_{#3}}\xspace}
\newcommand{\bft}[0]{BFT\xspace}
\newcommand{\hft}[0]{HFT\xspace}
\newcommand{\bftwv}[0]{BFT-WV\xspace}
\newcommand{\systemze}[0]{\mbox{\system-0E}\xspace}
\newcommand{\systemoe}[0]{\mbox{\system-1E}\xspace}

\newcommand{\headline}[1]{\vspace{1mm plus .5mm minus .5mm}\noindent\textbf{\textit{#1.}}~}

\usepackage{listings}
\usepackage{subscript}

\font\lsttt=rm-lmtl10 scaled 820
\font\lstbtt=rm-lmtk10 scaled 820

\newcommand{\commentsize}{\fontsize{8pt}{0pt}\selectfont}
\newcommand{\pseudocode}{\commentsize\normalfont}

\RequirePackage{substr}
\ifdefined\ifextended\else
\def\ifextended#1#2{\IfSubStringInString{\detokenize{_extended}}{\jobname}{#1}{#2}}
\fi

\makeatletter
\ifextended{
\settopmatter{printacmref=false}
\SelectFootnoteRule[2]{copyrightpermission}
\DeclareNewFootnote{altcopyright}
\renewcommand\footnotetextcopyrightpermission[1]{%
	\begingroup%
	\footnotetextaltcopyright{%
\parindent\z@\parskip0.1\baselineskip%
\vskip 1.5mm%
\setlength{\fboxsep}{1em}%
\newlength{\copylen}\setlength{\copylen}{\columnwidth}\addtolength{\copylen}{-2\fboxsep}%
\colorbox{yellow!70!black!20}{%
	\parbox{\copylen}{%
This is an extended version of the article
"Michael Eischer and Tobias Distler. 2020. Resilient Cloud-based Replication with Low Latency. In
\emph{\@acmBooktitle}\ifx\@acmDOI\@empty.\else, \@formatdoi{\@acmDOI}.\fi%
"
}}
\vspace{3.5mm}
}%
\endgroup%
}
}{}
\makeatother

\begin{document}

\title[Resilient Cloud-based Replication with Low Latency]{Resilient\hspace{2.5mm}Cloud-based\hspace{2.5mm}Replication\hspace{2.5mm}with\hspace{2.5mm}Low\hspace{2.5mm}Latency}
\ifextended{\subtitle{\vspace*{3mm}(Extended Version)}}{}


\author{Michael Eischer and Tobias Distler}
\affiliation{%
	\institution{Friedrich-Alexander University Erlangen-Nürnberg (FAU)}
	\city{Erlangen}
	\state{Germany\ifextended{\vspace*{1cm}}{}}
}




\ifextended{
	\fancyhead[LE]{}%
	\fancyhead[RO]{}%
}{}


\input{sections/abstract}
\maketitle

\input{sections/introduction}
\input{sections/background}
\input{sections/approach}
\input{sections/channel}
\input{sections/evaluation}
\input{sections/related}
\input{sections/conclusion}

\section*{Acknowledgments}

This work was partially supported by the German Research Council (DFG) under grant no. DI 2097/1-2~(``REFIT'').


\bibliographystyle{ACM-Reference-Format}
\bibliography{bib/paper-long}

\ifextended{
	\appendix
	\input{sections/proof}
}{}

\end{document}

%% file: sections/abstract.tex
\begin{abstract}

Existing approaches to tolerate Byzantine faults in geo\hyp{}replicated environments require systems to execute complex agreement protocols over wide-area links and consequently are often associated with high response times. In this paper we address this problem with \system, a resilient replication architecture for geo-distributed systems that leverages the availability characteristics of today's public-cloud infrastructures to minimize complexity and reduce latency. \system models a system as a collection of loosely coupled replica groups whose members are hosted in different cloud-provided fault domains~(i.e.,~availability zones) of the same geographic region. This structural organization makes it possible to achieve low response times by placing replica groups in close proximity to clients while still enabling the replicas of a group to interact over short-distance links. To handle the inter-group communication necessary for strong consistency \system uses a reliable group-to-group message channel with first-in-first-out semantics and built-in flow control that significantly simplifies system design.

\end{abstract}

%% file: sections/introduction.tex
\section{Introduction}

Byzantine fault-tolerant~(BFT) protocols enable a system to withstand arbitrary faults and consequently have been used to increase the resilience of a wide spectrum of \ifextended{}{\pagebreak} critical applications such as key-value stores~\mbox{\cite{padilha13augustus,padilha16callinicos,li16sarek,eischer19deterministic}}, SCADA systems~\cite{nogueira18challenges,babay18network,babay19deploying}, firewalls~\cite{bessani08crutial,garcia16sieveq}, coordination services~\cite{clement09upright,kapitza12cheapbft,behl15consensus,distler16resource,eischer19scalable}, and permissioned blockchains~\cite{sousa18byzantine,gueta19sbft}. To provide their high degree of fault tolerance, BFT protocols replicate the state of an application across a set of servers and rely on a leader-based consensus algorithm to keep these replicas consistent. This task requires several subprotocols~(e.g.,~for leader election, checkpointing, state transfer) and multiple phases of message exchange between replicas~\cite{castro99practical}.

Unfortunately, this complexity makes it inherently difficult to achieve low latency in use cases in which the clients of an application are scattered across various geo\-graphic locations. For example, placing replicas in close proximity to each other may reduce the latency of strongly consistent requests whose execution must be coordinated by the consensus protocol between replicas. However, with replicas being located farther apart from clients this strategy also increases the response times of requests such as weakly consistent reads that do not need to be agreed on and only involve direct interaction between clients and replicas. In contrast, co-locating replicas with clients has the inverse effect of speeding up client--replica communication but adding a significant performance overhead to the agreement protocol.

Existing approaches for BFT wide-area replication aim at minimizing this overhead by (1)~applying weighted-voting schemes to reduce the quorum sizes needed to complete consensus~\cite{sousa15separating,berger19resilient}, (2)~rotating the leader role among replicas to shorten the path necessary to insert a request into the agreement protocol~\mbox{\cite{veronese09spin,mao09towards,veronese10ebawa}}, or (3)~relying on a two-level system design that deploys an entire BFT replica cluster at each client site in order to be able to use crash-tolerant replication between sites~\cite{amir10steward,amir07customizable}. In all these cases, BFT systems still need to run complex consensus-based replication protocols over wide-area links which not only results in response-time overhead but also makes it difficult to dynamically introduce new \linebreak replica sites, for example, to serve clients at new locations.

In this paper we address these problems with \system, a cloud-based BFT system architecture for geo-replicated services that models a system as a collection of loosely coupled replica groups that are deployed in different regions. Separating agreement from execution~\cite{yin03separating}, one of the groups (``\emph{agreement group}'') establishes an order on all requests with strong consistency demands while all other groups~(``\emph{execution groups}'') are responsible for communicating with clients and processing requests. In contrast to existing approaches, \system does not require complex wide-area protocols but instead handles tasks such as consensus, leader election, and checkpointing within a group and over short-distance links. To make this possible while still offering resilience against replica failures, \system leverages the design of today's cloud infrastructures~\cite{ec2-regions,azure-regions,gce-regions} and places the replicas of a group in different availability zones of the same region; availability zones are hosted by data centers at distinct sites and specifically engineered to represent different fault domains.

In particular we make four contributions in this paper: (1)~We present the \system architecture and discuss how it achieves low latency for weakly consistent reads by placing execution groups close to clients, while at the same time minimizing agreement response times for strongly consistent reads and writes. (2)~We show how to design \system in a modular way so that execution groups do not depend on internals of the agreement group~(e.g.,~a specific consensus protocol). As an additional benefit, the modularity also makes it straightforward to add/remove execution groups at runtime. (3)~We introduce a wide-area BFT flow-control mechanism that exploits the special characteristics of \system to minimize complexity. Our approach is based on a simple message-channel abstraction that handles the inter-regional communication between two replica groups and prevents one group \linebreak from overwhelming the other. (4)~We evaluate \system in comparison to the state of the art in BFT wide-area replication.


%% file: sections/background.tex
\section{Background and Problem Statement}
\label{sec:background}

In this section, we present background on existing approaches and common requirements of BFT wide-area replication.

\subsection{System Model}

Our work focuses on stateful applications with strong reliability requirements whose clients are scattered across different geographic locations. To access the application a client submits a request to the server side. We assume that both clients and servers can be subject to Byzantine faults. As a consequence, nodes~(i.e.,~clients and servers) do not trust each other and do not make irreversible decisions based on the input provided by another node alone. For example, to tolerate up to $f$~faulty servers, a client only accepts a result after it has \linebreak obtained at least $f+1$~matching replies from different servers.

Besides service availability and correctness in the presence of failures, low latency is a primary concern in our target systems. Achieving this goal while keeping the states of servers consistent is inherently difficult in use cases in which clients are geographically dispersed. The problem is further complicated by the fact that we assume that the locations from which clients access the application may change over time, typically as a result of the global day/night cycle. To continuously provide low latency under such conditions, a system must offer some kind of reconfiguration mechanism enabling an adaptation to varying workloads. One possibility to achieve this, for example, is to dynamically include additional \linebreak servers that are located closer to newly started clients.

\subsection{Existing Approaches}
\label{sec:background-approaches}

In the following, we elaborate on the problems associated with Byzantine fault tolerance in geo-distributed systems and discuss existing approaches to solve them.


\begin{figure}
	\vspace{-3mm}
 	\subfloat[PBFT~\cite{castro99practical}\label{fig:pbft}] {
		\includegraphics{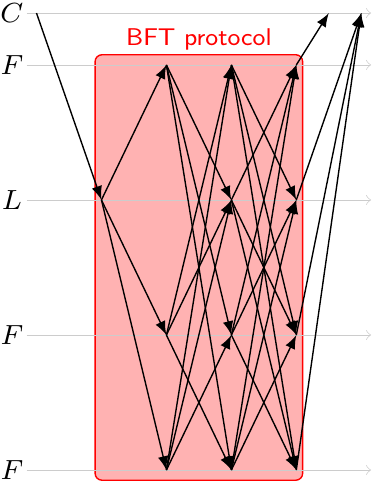}
	}
	\hfill
	\subfloat[Steward~\cite{amir10steward}\label{fig:steward}] {
		\includegraphics{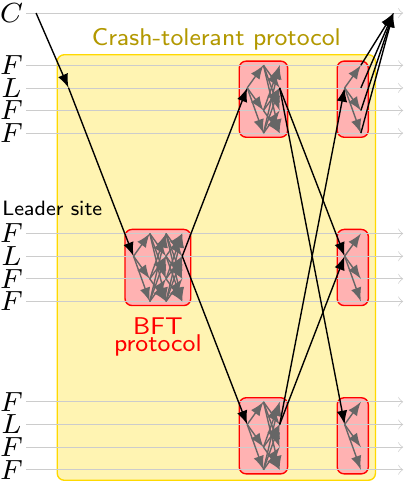}
	}
	\caption{System architectures for BFT geo-replication connecting a client~(C) with leader~(L) and follower~(F) replicas.}
\end{figure}

\headline{BFT in Wide-Area Environments}
%
The straightforward approach to offer resilience against arbitrary failures is to rely on a BFT replication protocol, for example PBFT~\cite{castro99practical}. As illustrated in Figure~\ref{fig:pbft}, PBFT requires at least $3f+1$~replicas to tolerate $f$~failures. To keep the application state consistent across replicas, PBFT ensures that replicas run an agreement protocol to decide in which order to process client requests. For this purpose, PBFT elects one of the replicas as leader~(marked $L$ in Figure~\ref{fig:pbft}) while all other replicas assume the roles of followers~($F$). Having received a new request, the leader is responsible for initiating the agreement process, which then involves multiple message exchanges between replicas. To deal with scenarios where a faulty leader does not behave according to specification, for example by ignoring a request, PBFT provides a mechanism that enables followers to depose the leader and appoint a new one. Once the agreement process is complete, all non-faulty replicas execute the request and send the result to the client, thereby enabling the client to validate the result by comparison.

Using BFT protocols such as PBFT to build resilient systems is effective but has several disadvantages in the context of geo-replication: (1)~With replicas being distributed across different geographic sites, the entire BFT~protocol needs to be executed over wide-area links, which often results in high response times. Note that this is not only true with regard to the task of agreeing on requests during normal operation, but for example also for electing a new leader as part of fault handling. (2)~Due to the fact that all requests must flow through the leader, the geographic location of the leader, and in particular its position relative to the majority of followers, usually has a significant influence on latency~\mbox{\cite{sousa15separating,eischer18latency}}. Consequently, a leader switch may decisively change a system's performance characteristics, requiring clients to deal with the associated latency volatility. (3)~Consisting of only $3f+1$~replicas, for traditional BFT systems it is inherently difficult to select suitable replica locations in cases where a large and varying number of clients are scattered across the globe. Ideally, replicas would be placed both in close distance to each other (to speed up agreement) as well as in close distance to clients (to minimize the transmission time of requests and results). For systems with just a few replicas but many \linebreak clients meeting this requirement is essentially impossible.

\headline{Weighted Voting}
%
By assigning different weights on the votes replica have within the consensus protocol~\cite{sousa15separating,berger19resilient} it is feasible to introduce additional replicas while keeping response times low or even reducing them in a geo-replicated setting. Unfortunately, this comes at the cost of an increased number of messages exchanged between replicas, which can be prohibitively expensive in public-cloud settings as providers typically charge extra for wide-area traffic.

\headline{Leader Rotation}
%
Different authors have proposed to improve performance by rotating the leader role among replicas, following the idea of enabling each client to submit requests to its nearest replica~\mbox{\cite{veronese09spin,mao09towards,veronese10ebawa}}. Results from an extensive experimental evaluation by Sousa~et~al.~\cite{sousa15separating}, however, showed that in practice this approach does not provide significant benefits compared with appointing a fixed leader at a well-connected site. Besides, leader rotation still requires the execution of a complex protocol over wide-area links.

\headline{Hierarchical System Architecture}
%
To increase the scalability of BFT systems in wide-area settings, Amir et al. presented a hierarchical architecture as part of their Steward system~\cite{amir10steward}. As shown in Figure~\ref{fig:steward}, instead of hosting a single replica, each site in Steward comprises a cluster of replicas that run a site-local BFT agreement protocol. A key benefit of this approach is the fact that, although individual replicas still may be subject to Byzantine faults, an entire cluster can be assumed to only fail by crashing. This property at the local level enables Steward to rely on a crash-tolerant agreement protocol at the global level~(i.e.,~between sites), which compared with traditional BFT systems requires fewer phases and fewer message transmissions over wide-area links.

The efficiency enhancements made possible by its architecture enable Steward to improve performance, however, they come at the cost of an increased overall complexity that stems from the need to maintain replication protocols at two levels: within each site as well as between sites. Designing and implementing such protocols in isolation already is a non-trivial task, additionally guaranteeing a correct interplay between them is even more challenging. To ensure liveness Steward, for example, requires timeouts at different levels to be carefully coordinated~\cite{amir10steward}. Amir et al. addressed these problems in a subsequent work~\cite{amir07customizable}, which in this paper we refer to as CFT-WAR. In contrast to Steward, in CFT-WAR each step of the wide-area protocol~(e.g.,~Paxos~\cite{lamport98part}) is handled by a full-fledged multi-phase consensus protocol at each site~(e.g.,~PBFT). As a main advantage, this approach disentangles the protocols used for wide-area and site-internal replication. On the downside, it introduces additional overhead that in general prevents CFT-WAR from achieving response times as low as Steward's when providing the same degree of fault tolerance~\cite{amir07customizable}. Furthermore, due to performing agreement at two levels CFT-WAR still needs to run multiple subprotocols for tasks such as leader election, one at each level.
A set of additional subprotocols would be required to support the dynamic addition/removal of individual replicas or entire sites in a hierarchical system architecture, thereby further increasing complexity. To our knowledge, the ability to adjust to varying workload conditions was not a design goal of Steward and CFT-WAR, which is why the systems do not offer \linebreak mechanisms for changing their composition at runtime.

\subsection{Problem Statement}

Our analysis in Section~\ref{sec:background-approaches} shows that applying existing approaches to provide BFT in a cloud-based geo-replicated environment is possible, for example with regard to safety, but cumbersome due to the associated high complexity and the lack of effective means to react to changing workloads. This observation led us to ask whether these problems can be circumvented by a BFT system architecture that is specifically tailored to the characteristics of today's cloud infrastructures. In particular, we aim for a resilient system architecture that has three properties: efficiency, modularity, and adaptability.

\headline{Efficiency}
%
To minimize response times during both normal-case operation as well as fault handling, a system architecture in the ideal case does not require the execution of complex protocols over wide-area links. Instead, tasks involving multiple phases of message exchange between replicas, such as the agreement on requests, should be handled by replicas that are located in comparably close distance to each other.

\headline{Modularity}
%
Supporting a variety of cloud use cases with different requirements is difficult if the protocols responsible for the agreement and execution of requests are hard-wired into the BFT system architecture. To address this issue, we join other authors~\cite{amir07customizable} in aiming for an architecture that, for example, can be integrated with different consensus protocols depending on the specific demands of an application.

\headline{Adaptability}
%
One major strength of public clouds is to quickly provide resources on demand and at various geographic locations all over the globe. A BFT system architecture should be able to leverage this feature for hosting replicas in the proximity of clients to reduce the latency with which clients access the replicated service. Specifically, if new clients are started at other sites, there should be a lightweight mechanism for dynamically adding new replicas. The same applies to means for removing replicas that are no longer of \linebreak benefit as the clients in their vicinity have been shut down.


%% file: sections/approach.tex
\section{\system}

This section presents the cloud-based BFT system architecture \system. In particular, we focus on how the architecture achieves low latency by performing consensus only over short-distance links, how \system achieves modularity by relying on a novel message-channel abstraction, and how it can be dynamically reconfigured to adapt to workload changes.

\subsection{Architecture}

Targeting use cases in wide-area environments, \system's system architecture is distributed across multiple geographic sites. For this purpose, \system leverages the common organizational structure of state-of-the-art cloud infrastructures such as Amazon EC2~\cite{ec2-regions}, Microsoft Azure~\cite{azure-regions}, or Google Compute Engine~\cite{gce-regions} by grouping sites into \emph{regions}, as shown in Figure~\ref{fig:architecture}. The sites within a region typically are several tens of kilometers apart from each other and represent separate fault domains, commonly referred to as \emph{availability zones}. In addition to constructing the data centers at distinct geographic locations, cloud providers also ensure that data centers in different availability zones are equipped with dedicated power supply systems and network links to minimize the probability of dependent failures. For the \system system architecture, availability zones play an important role as they allow us to place replicas in separate fault domains and still enable them to interact over short-distance links with comparably low latency.

\headline{Replica Groups}
%
Relying on this setting, \system is composed of multiple loosely coupled replica groups, each being distributed across different availability zones of a specific region. One of the replica groups in the system, the \emph{agreement group}, is responsible for establishing a global total order on incoming requests. The size of this group depends on the protocol it uses for consensus. Running PBFT~\cite{castro99practical}, for example, the agreement group consists of $3f_a+1$~replicas and is able to tolerate $f_a$~Byzantine faults. All other replica groups in the system, the \emph{execution groups}, host the application logic, process the ordered requests, and handle the communication with clients. Each of these groups comprises $2f_e+1$~replicas and tolerates at most $f_e$~Byzantine faults.
The level of fault tolerance provided by the agreement group and the executions groups may be selected independently. Supporting multiple execution groups enables \system to scale throughput by adding/removing groups and to minimize latency by placing groups in the vicinity of clients.

\headline{Execution-Replica Registry}
%
\system contains an execution-replica registry to provide clients with information on the locations and addresses of active replicas. The registry is a BFT service that is hosted and maintained by the agreement group. Its contents are updated by agreement replicas whenever the composition of the system changes~(see Section~\ref{sec:adaptability}).

\headline{Efficient BFT Replication}
%
In contrast to existing ap\-proach\-es~(see Section~\ref{sec:background-approaches}), \system does not run a full-fledged and complex replication protocol over long-distance links. Instead, all non-trivial tasks~(e.g.,~reaching consensus on requests) are carried out within a replica group using low-latency intra-region connections. Following this design principle, \system handles requests by forwarding them along a chain of stages represented by different replica groups. Specifically, clients submit their requests to their nearest execution group, which in turn forwards the request to the agreement group for ordering. Once this step is complete, the agreement group instructs all execution groups to process the ordered request.
This ensures that execution-group states remain consistent without requiring the execution groups to reach consensus themselves. Having processed the request, the replicas of the execution group the client is connected to return the result. As each execution group comprises $2f_e+1$~replicas, clients are able to verify the correctness of a result solely based on the replies they receive from their local execution group.

With all communication-intensive steps being performed over intra-region links, inter-region links in \system are only responsible for forwarding the outputs of one stage to the replica group(s) constituting the next stage. In particular, this approach has the following benefits: (1)~It greatly simplifies the interaction of replicas over long-distance connections. (2)~It enables a modular design that allows different deployments to rely on different agreement protocols without the need to modify the implementation of execution replicas. (3)~As we show in Section~\ref{sec:channels}, it allows \system to use the same abstraction, a reliable message channel, for all inter-region links, thereby facilitating system implementation.


\begin{figure}
	\includegraphics{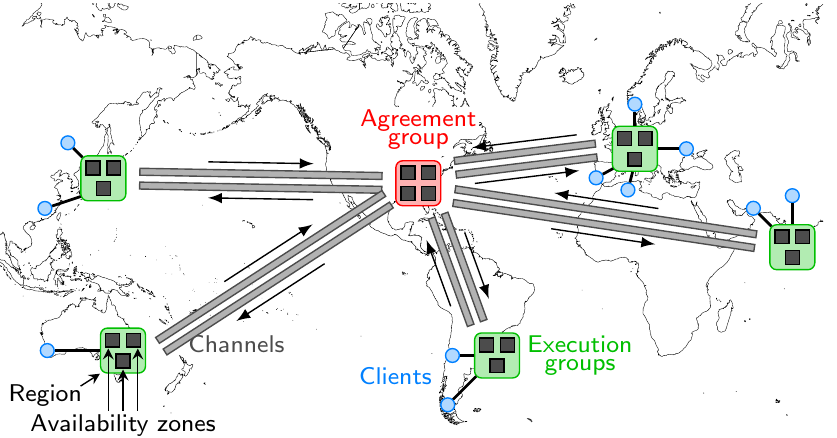}
	\caption{\system system architecture}
	\label{fig:architecture}
\end{figure}

\headline{Practical Considerations}
%
As of this writing, all major public clouds offer several regions with at least three availability zones~(Amazon EC2:~20, Microsoft Azure:~10, Google Compute Engine:~24) and therefore support the world-wide deployment of \system execution groups which tolerate one faulty replica. In addition, Amazon~(Virginia, Oregon, Tokyo) and Google~(Iowa) also already operate regions with four or more availability zones, which consequently are candidates for hosting \system's agreement group. With public cloud infrastructures still being expanded, new regions and availability zones are added every year, increasing the deployment options for \system. Besides, to further improve the resilience of \system, agreement and execution replicas may be distributed across different clouds, thereby reducing the dependence on a single provider~\cite{bessani13depsky,abu-libdeh10racs}. As there are several regions hosting data centers and availability zones of multiple providers~(e.g.,~Europe, North America, South America, India, Asia, and Australia), this approach also makes it possible to deploy larger agreement and execution groups that tolerate $f_a>1$ and $f_e>1$ replica failures, respectively.

Representing distinct fault and upgrade domains, availability zones are designed to enable uninterrupted execution of services that are replicated within the same region. Despite the efforts undertaken by providers, in the past there have been rare incidents where problems in one availability zone caused temporary availability issues in other zones belonging to the same region~\cite{aws11incident}. In \system, if more than $f_a$~agreement replicas are unresponsive, the agreement group temporarily cannot order new requests until the replicas become available again. However, as we detail in Section~\ref{sec:protocol}, in such cases \system is still able to process weakly consistent read requests as these operations are handled within a client's local execution group. On the other hand, if more than $f_e$~replicas of the same execution group become unavailable, affected clients can temporarily switch to a different execution group and continue to use the service.

\subsection{Inter-Regional Message Channels}
\label{sec:channels}

To support a modular design, we use an abstraction to handle all interaction between replica groups in \system: the \emph{inter-regional message channel~(\channel)}. Specifically, \channel{}s are responsible for forwarding messages from a group of sender replicas in one region to a group of receiver replicas in another region. Conceptually, \channel{}s can be viewed as an extension of BLinks~\cite{amir07customizable}, however, unlike BLinks, \channel{}s (1)~do not require messages to be totally ordered at the channel level and (2)~comprise built-in flow control. To forward information, an \channel internally can be divided into multiple subchannels providing first-in-first-out semantics. Each subchannel has a configurable maximum capacity~(i.e.,~an upper bound on the number of messages that can be concurrently in transmission) and relies on a window-based flow-control mechanism to prevent senders from overwhelming receivers. Below, we discuss the specifics of \channel{}s at a conceptual level. For possible implementations please refer to Section~\ref{sec:implementations}.

\headline{Overview}
%
Figure~\ref{fig:channel} presents an example \channel that comprises two subchannels and connects four senders to three receivers. Subchannels of the same \channel are independent of each other and can be regarded as distributed queues with limited capacity that distinguish messages based on unique position indices. Both senders and receivers run dedicated endpoints which together form the \channel and enable the replicas to access it. When a replica sends a message, it provides its local endpoint with the information which subchannel and position to use for the message~(\texttt{send()}). Similarly, to receive a message a replica queries its local endpoint for the message corresponding to a specific subchannel and position~(\texttt{receive()}). In addition, \channel endpoints offer a method to shift the flow-control window of a subchannel~(\texttt{move\_window()}), as further discussed below.

\headline{Send Semantics}
%
\channel{}s are not designed to exchange arbitrary messages between replicas but instead provide specific send semantics enabling \system to safely forward the decision of a replica group to another. In particular, tolerating at most $f_s$~senders with Byzantine faults, the \channel only forwards a message after at least $f_s+1$~different senders transmitted a message with identical content using the same subchannel and position. Consequently, in order for a message to pass the channel at least one correct sender must have vouched for the validity of the message's content and requested its transmission. In contrast, messages solely submitted by the up to $f_s$~faulty senders have no possibility of getting through and being delivered to receivers.


\begin{figure}
	\vspace{-3mm}
	\begin{lstlisting}
%\hrule\vspace{1mm}\begin{center}\lstbtt{Sender Endpoint Interface}\end{center}\hrule\vspace{1.5mm}%
void %\lstbtt{send}%(%\textsc{Subchannel}% sc, %\textsc{Position}% p, %\textsc{Message}% m);
void %\lstbtt{move\_window}%(%\textsc{Subchannel}% sc, %\textsc{Position}% p);
%\vspace{-2mm}\hrule\vspace{1mm}\begin{center}\lstbtt{Receiver Endpoint Interface}\end{center}\hrule\vspace{1.5mm}%
%\textsc{Message}% %\lstbtt{receive}%(%\textsc{Subchannel}% sc, %\textsc{Position}% p);
void %\lstbtt{move\_window}%(%\textsc{Subchannel}% sc, %\textsc{Position}% p);
%\vspace{-2mm}\hrule%
	\end{lstlisting}
	\vspace{1mm}
	\includegraphics{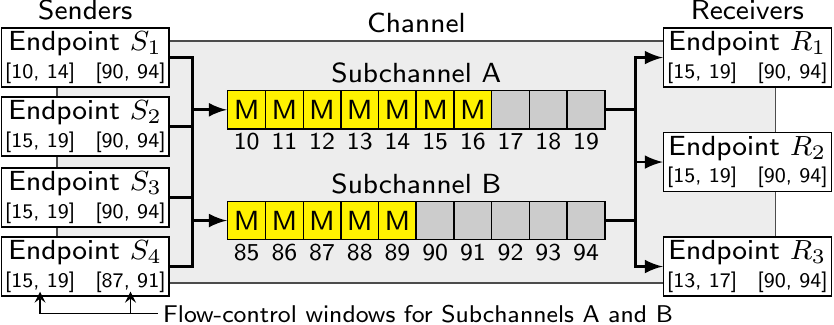}
	\caption{Conceptual view of an example \channel with two independent subchannels that both have a maximum capacity of ten messages~(M). Senders~($S_*$) and receivers~($R_*$) access the subchannels via their local endpoints; each endpoint manages its own subchannel-specific flow-control windows.}
	\label{fig:channel}
	\label{fig:interface}
\end{figure}

\headline{Authentication}
%
\channel{}s protect all channel-internal communication with digital signatures to enable the recipient of a message to verify the integrity and the origin of the message. If an endpoint is unable to validate the authenticity of a received message, it immediately discards the message.

\headline{Flow Control}
%
With the capacities of subchannels being limited, \channel endpoints apply a flow-control mechanism to coordinate senders and receivers.
For this purpose, for each subchannel an endpoint manages a separate window that restricts which messages a sender/receiver is able to transmit/obtain at a given time.
If a subchannel's window at a sender endpoint is full, the sender cannot insert additional messages into this subchannel until the endpoint moves the window forward.
In the normal case, this action is triggered by receivers calling \texttt{move\_window()} and requesting the start of the window to be shifted to a higher position.
Whenever a sender endpoint learns that the window position has changed at one of the receiver endpoints, the sender endpoint sets its own window start to the $f_r+1$~highest position requested by any receiver where $f_r$~denotes the number of receivers with Byzantine faults to tolerate.
This ensures that correct sender endpoints only move their windows, and thus discard messages at lower positions, after receiving the information that at least one correct receiver has permitted such a step.

Besides receiver-driven window shifts, our channels also allow senders to request an increase of the starting position of a subchannel's window. If senders opt to do so, it may become impossible for a receiver endpoint to provide the message at the position the endpoint's local replica requested. The same scenario can occur if a receiver endpoint is slow or falls behind~(e.g.,~due to a network problem) while \mbox{$f_r+1$}~other receivers have already requested the window to be moved forward. In such cases, the affected receiver endpoint aborts the \texttt{receive()} call with an exception and thereby enables its local replica to handle the situation. As discussed in Section~\ref{sec:checkpointing}, replicas react to such an exception by obtaining the missed information from other replicas.

\headline{Use in \system}
%
\channel{}s are an essential building block of \system's modular architecture as they enable us to design a geo-replicated BFT system as a composition of loosely coupled replica groups that interact using the same channel abstraction. In particular, \system relies on two different \channel instances to perform all inter-group communication over long-\linebreak distance links: the \textit{request channel} and the \textit{commit channel}.

The request channel allows an execution group to forward newly received requests to the agreement group; that is, this channel is an \channel that connects $2f_e+1$~senders~(i.e.,~execution replicas) to $3f_a+1$~receivers~(i.e.,~agreement replicas). To transmit the requests, the request channel comprises multiple subchannels, one for each client. In contrast, the commit channel only consists of a single subchannel and is used by the agreement group to inform an execution group about the totally ordered sequence of agreed requests. The commit channel consequently is responsible for forwarding the decisions of $3f_a+1$~senders to $2f_e+1$~receivers. In summary, the agreement group maintains a pair of \channel{}s~(i.e.,~one request channel and one commit channel) to each execution group.

\subsection{Request Handling}
\label{sec:protocol}
%
\system differentiates between requests that potentially modify application state~(``writes'') and those that do not~(``reads''). This distinction enables the system to handle requests of each category as efficiently as possible. While writes need to be applied to all execution groups to keep their states consistent, it is sufficient for reads to only process them at the execution group a client is connected to. Figure~\ref{fig:protocol} gives an overview of how requests flow through \system. Below, we provide details on the system's replication protocols for writes and reads. In this context, it is important to note that all messages exchanged between clients and replicas must be authenticated, for example using HMACs~\cite{tsudik92message}. For messages sent through \channel{}s, the authentication is handled by the channels.

In the following, we describe \system's handling of write and read requests . The proof of correctness and liveness is deferred to the \ifextended{appendix of the paper}{extended version of the paper~\cite{eischer20resilient-extended}}.

\begin{figure}
	\includegraphics{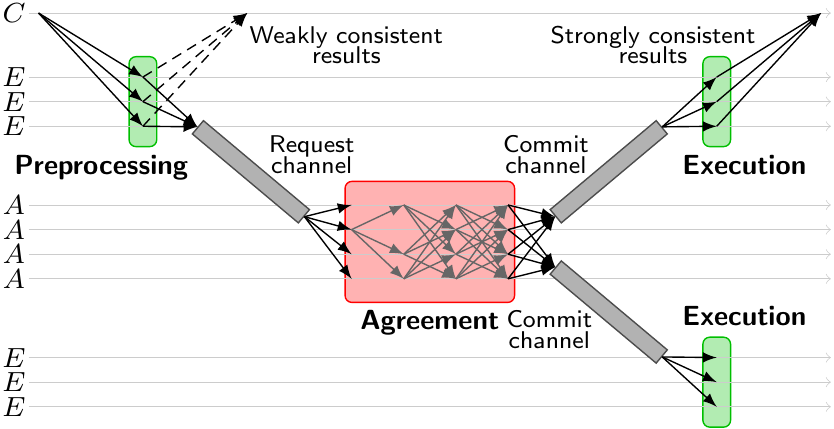}
	\caption{Overview of \system's replication protocol}
	\label{fig:protocol}
\end{figure}

\headline{Writes}
%
\system's protocol for writes is presented in Figure~\ref{def:pseudo-write}.
To perform a write operation~$w$, a client~$c$ creates a corresponding message \msg{Write}{w, c, t_c} using a unique client-local counter value~$t_c$ and sends the message to all replicas of an execution group. In general, a client for this purpose may select any execution group in the system, however, in an effort to minimize latency, \system clients typically choose the group closest to their own site.

When an execution replica receives the client's request, it first checks whether the message is correctly authenticated and whether the client has permission to access the system. If any of these checks fail the replica discards the message. Otherwise, the replica of execution group~$e$ wraps the entire request~$r$ in a message \msg{Request}{r, e} and submits the message to the agreement group via its request channel. More precisely, unless the execution replica  has already forwarded the request~(Lines~\ref{code:exec-cache-start}--\ref{code:exec-cache-end}) it moves the window of the client's subchannel to position~$t_c$ and inserts the write request at that position~(L.~\ref{code:exec-send-start}--\ref{code:exec-send-end}). Once at least $f_e+1$~members of the execution group~(i.e.,~at least one correct execution replica) have validated and forwarded the request, the request channel permits agreement replicas to retrieve the message~(L.~\ref{code:ag-receive}). This allows the agreement group to initiate the consensus process for the message~(L.~\ref{code:ag-order}), which is then performed entirely within the group's region. Having learned that the request is committed and has been assigned the agreement-sequence number~$s$~(L.~\ref{code:ag-deliver}), an agreement replica creates a confirmation \msg{Execute}{r, s}. As write operations need to be processed by all execution groups, the agreement replica sends this message through all commit channels at position~$s$~(L.~\ref{code:ag-execute}).

Once $f_a+1$~agreement replicas~(among them at least one correct replica) have sent an \textsc{Execute} message with the same content and sequence number, a commit channel enables its receivers to obtain the message~(L.~\ref{code:exec-receive}). Having done so, an execution replica processes the included request by applying the corresponding write to its local state~(L.~\ref{code:exec-exec-start}). Each replica of execution group~$e$ also returns a reply \msg{Result}{u_c, t_c} with the operation's result~$u_c$ to the client that submitted the request with counter value~$t_c$~(L.~\ref{code:exec-exec-end}). The client accepts a result after it has received $f_e+1$~replies with matching result and counter value from different execution replicas.

As we detail in Section~\ref{sec:checkpointing}, when processing writes replicas in \system also create periodic checkpoints~(L.~\ref{code:exec-cp-gen-start}--\ref{code:exec-cp-end} and \ref{code:ag-cp-gen-start}--\ref{code:ag-cp-end}) to assist other replicas that might have fallen behind.


{%
\font\lstttx=rm-lmtl10 scaled 820
\font\lstbttx=rm-lmtk10 scaled 820
\lstset{
	basicstyle=\linespread{.925}\footnotesize\lstttx,
	emphstyle=\lstbttx,
	commentstyle=\commentsize\textit,
	tabsize=2,
	numberstyle=\scriptsize,
	numbersep=2.2mm,
	xleftmargin=5mm,
	numbers=left,
	frame=none,
	columns=fullflexible,
	numberblanklines=true,
	emptylines=2,
	breaklines=true,
	breakatwhitespace=false,
	escapechar=\%,
	mathescape,
	morecomment=[l]{//},
	morecomment=[s]{/*}{*/},
	morestring=[b]",
	aboveskip=-3mm,
	emph={%
		until, to, sleep, for, on, receive, from, if, send, main, loop, while, else, parallel, each, move_window%
	}%
}%
\begin{figure}
	\vspace{2mm}\hrule\pseudocode{\begin{center}Execution Replica of Execution Group $e$\end{center}}\hrule\vspace*{2mm}\vspace{8pt}
	\begin{lstlisting}
%$s_n$% := 0%\hfill%// Current sequence number
%$t[c]$% := 0%\hfill%//Vector with counter value of latest forwarded client request
%$u[c]$% := %$\varnothing$\hfill%// Reply cache %$\langle\textsc{Reply}, u_c, t_c\rangle$%%\vspace{2pt}\hrule{}\vspace{3pt}%
on receive(%$r$% = %$\langle\textsc{Write},w,c,t_c\rangle$% from %$c$%):
	if %$t_c \leq t[c]$%: return send result %$u[c]$% to %$c$%%\label{code:exec-cache-start}%
	%$t[c]$% := %$t_c$\hfill%// Remember forwarded request%\label{code:exec-cache-end}%
	request-IRMC.move_window(%$c$%, %$t_c$%) %\label{code:exec-send-start}%
	request-IRMC.send(%$c$%, %$t_c$%, %$\langle\textsc{Request},r,e\rangle$%)%\label{code:exec-send-end}%%\vspace{5pt}%
main loop:
	%$m$% := commit-IRMC.receive(0, %$s_n+1$%)%\label{code:exec-receive}%
	if %$m = \langle \textsc{TooOld}, s'\rangle$%: %\texttt{fetch checkpoint for} $s'$%
	else:
		%$m$% = %$\langle\textsc{Execute}, \langle\textsc{Request},\langle\textsc{Write}, w, c, t_c\rangle,e'\rangle, s_n+1\rangle$%
		%$u_c$% := %\texttt{Application execute}% %$m$%%\label{code:exec-exec-start}%
		%$s_n$% := %$s_n$% + 1
		send %$\langle\textsc{Result}, u_c, t_c\rangle$% to %$c$% if %$e'$\,%=%\,$e$% and store in %$u[c]$%%\label{code:exec-exec-end}%
		if %$s_n \equiv 0$ mod $k_e$%:%\label{code:exec-cp-gen-start}%
			%\texttt{create checkpoint for}% %$s_n$% %\texttt{with}% %$u$% and %\texttt{Application}%%\vspace{5pt}%
on %\texttt{stable checkpoint}%(%\textsc{SeqNr} $s$%, %$u'$%, %\texttt{Application}%'):%\label{code:exec-cp-start}%
	commit-IRMC.move_window(0, %$s+1$%)
	if %$s \geq s_n$%: %\texttt{apply checkpoint to}% %$s_n$%, %$u$% and %\texttt{Application}%%\label{code:exec-cp-end}%%\vspace{3mm}\hrule\pseudocode{\begin{center}Agreement Replica\end{center}}\hrule\vspace*{2mm}%
%$s_n$% := 0%\hfill%// Current sequence number
%$t[c]$% := 0%\hfill%// Counter values of latest agreed request; used by consensus
%$t^+[c]$ := 0\hfill%// Counter values for next expected request
AG-WIN%\,$\geq$\,$k_a$%%$\hfill$%// Size of agreement window
win := [1,AG-WIN]%\hfill%// Range with [lower, upper] bound (inclusive)
hist := last %$|$%commit-IRMC window%$|$ \textsc{Execute} messages%%\vspace{2pt}\hrule{}\vspace{3pt}%
for each client %$c$% and execution group %$e$% in parallel:
	while true:
		%$m$% := request-IRMC.receive(%$c$%, %$t^+[c]$%) from group %$e$%%\label{code:ag-receive}%
		if %$m$% = %$\langle\textsc{TooOld}, t_c\rangle$%: %$t^+[c]$% := %$t_c$%%\label{code:ag-receive-tooold}%
		else: // %$m$% = %$\langle\textsc{Request},\langle\textsc{Write}, w, c, t_c\rangle,e\rangle$%
			%\texttt{Consensus order request} $m$%%\label{code:ag-order}%
			%$t^+[c]$% := %$t^+[c] + 1$%%\vspace{5pt}%
on %\texttt{Consensus ordered}%(%\textsc{SeqNr}~$s$, $r$% = %$\langle\textsc{Request},\langle\textsc{Write},w,c,t_c\rangle,e\rangle$%):%\label{code:ag-deliver}%
	sleep until upper limit of win %$> s$%%\label{code:ag-win-sleep}%
	%$s_n$% := %$s$%
	%$t[c]$% := %$t_c$%
	%$t^+[c]$% := %$\max(t_c + 1, t^+[c])$%
	commit-IRMC.send(0, %$s$%, %$\langle\textsc{Execute}, r, s\rangle$%) for each execution group %$e$% and add %\textsc{Execute}% to hist%\label{code:ag-execute}%
	if %$s_n \equiv 0$ mod $k_a$%:%\label{code:ag-cp-gen-start}%
		%\texttt{create checkpoint for}% %$s_n$% %\texttt{with}% %$t$%, hist%\label{code:ag-cp-gen-end}%%\vspace{5pt}%
on %\texttt{stable checkpoint}%(%\textsc{SeqNr} $s$%, %$t'$%, hist'):
	commit-IRMC.move_window(0, %$s-|$hist'$|+1$%)
	%\texttt{Consensus collect garbage before} $s+1$%%\label{code:ag-cp-gc}%
	if %$s > s_n$%:
		h_missing := %$\{\langle\textsc{Execute}, r, s'\rangle \in$% hist'$\,|\,s' \in [s_n+1, s]\}$
		%\texttt{apply checkpoint to}% %$s_n$%, %$t$% and hist%\label{code:ag-cp-apply-start}%
		for each execution group %$e$%:
			send h_missing via commit-IRMC of group %$e$\label{code:ag-cp-apply-end}%
	win := [%$s$%+1, %$s$%+AG-WIN]%\label{code:ag-win-move}%%\label{code:ag-cp-end}%
	\end{lstlisting}
	\hrule
	\caption{\system protocol for writes (pseudo code)}
	\label{def:pseudo-write}
\end{figure}
}%

\headline{Reads}
%
For reads, \system offers two different operations providing weakly consistent and strongly consistent results, respectively. To perform a weakly consistent read, a client sends a read request to all members of an execution group, which for a valid request immediately responds with a result, as illustrated by the dashed lines in Figure~\ref{fig:protocol}. As for writes, a client verifies the result based on $f_e+1$ matching replies. Weakly consistent reads achieve low latency as they only involve communication between the client and its execution group. Due to these reads being processed without further coordination with writes, in the presence of concurrent writes to the same state parts they may return stale values or fewer than $f_e+1$ matching results, similar to the optimized reads in existing BFT protocols~\cite{castro99practical,sousa15separating}. \system clients react to stalled \linebreak reads by retrying the operation or performing a strongly consistent read, which is guaranteed to produce a stable result.

Strongly consistent reads in \system for the most part have the same control and data flow as writes, with one important exception. With reads not modifying the application state, it is sufficient to process them at the client's execution group. Consequently, after a read request completed the consensus process, agreement replicas only forward it to the execution group that needs to handle the request. The \textsc{Execute}s to all other groups instead contain a placeholder including only the client request counter value for the same sequence number, thereby minimizing network and execution overhead.

\subsection{Checkpointing}
\label{sec:checkpointing}

As discussed in Section~\ref{sec:channels}, an \channel{} may garbage-collect messages before they have been delivered to all correct receivers. In the normal case in which all receivers advance at similar speed, this property usually does not take effect, resulting in each receiver to obtain every message. To address exceptional cases in which a correct receiver misses messages~(e.g.~due to a network problem), \system provides means to bring the affected receiver up to date via a checkpoint. The specific contents of a checkpoint vary depending on the receiver-replica group~(see below).
Checkpoints are periodically created after a group has agreed on\,/\,processed the message for a sequence number~$s$ that satisfies $s \equiv 0~mod~k$. The checkpoint interval~$k$ of a replica group is configurable and for the execution to sustain liveness must be smaller than the maximum capacity of the group's input \channel. The agreement-checkpoint interval~$k_a$ may be selected independently from the interval for execution checkpoints~$k_e$.

\headline{Agreement Checkpoints}
%
Having completed the consensus process for a request for which a checkpoint is due, an agreement replica creates an agreement snapshot and includes (1)~a vector~$t$ that for each client contains the counter value~$t_c$ of the client's latest agreed request and (2)~the last \textsc{Execute} messages corresponding to the commit subchannel capacity~(L.~\ref{code:ag-cp-gen-start}--\ref{code:ag-cp-gen-end} in Figure~\ref{def:pseudo-write}). In a next step, the agreement replica computes a hash~$h$ over the snapshot and sends a message \msg{Checkpoint}{h, s} protected with a digital signature to all members of its group. Having obtained $f_a+1$~correctly signed and matching checkpoint messages for the same sequence number, a replica has proof that its snapshot is correct. At this point, the replica can move forward its separate window used to ensure the periodic creation of a new checkpoint~(L.~\ref{code:ag-win-sleep} and \ref{code:ag-win-move}) and also instruct the consensus protocol to garbage collect preceding consensus instances~(L.~\ref{code:ag-cp-gc}).

Agreement replicas require periodic checkpoints to continue ordering new requests and thus there is at least one correct agreement replica that possesses both a corresponding valid checkpoint as well as proof of the checkpoint's correctness in the form of $f_a+1$~matching checkpoint messages.
As a consequence, if a correct agreement replica falls behind and queries its group members for the latest checkpoint, the replica will eventually be able to acquire this checkpoint, verify it, and apply it in order to catch up by skipping consensus instances. In such case, the checkpoint enables the replica to learn (1)~the request-subchannel positions at which to query the \channel for the next client requests and (2)~the \textsc{Execute}s of the skipped consensus instances~(L.~\ref{code:ag-cp-apply-start}--\ref{code:ag-cp-apply-end}).

\headline{Execution Checkpoints}
%
Execution-group checkpointing follows the same basic work flow as in the agreement group. An execution snapshot comprises a copy of the application state and the latest reply to each client, similar to the checkpoints in Omada~\cite{eischer19scalable}. This information enables a trailing execution replica to consistently update its local state without needing to process all agreed requests. When an execution checkpoint for a sequence number~$s$ becomes stable at an execution replica, the replica moves the flow-control window of its incoming commit channel to $s+1$~(L.~\ref{code:exec-cp-start}--\ref{code:exec-cp-end}).
This ensures that agreed requests are only discarded after at least one correct execution replica has collected a stable checkpoint.
Note that there is no need for checkpoints to contain requests.
A client moves its request subchannel's window forward by issuing a new request, thereby confirming that the old request \linebreak can be garbage-collected from the \channel.
This also allows execution replicas to skip forward to the current request~(L.~\ref{code:ag-receive-tooold}).

\subsection{Global Flow Control}
\label{sec:flow}

With the flow-control mechanism of an \channel only operating at the communication level between two replica groups, \system takes additional measures to coordinate the message flow at the point where the endpoints of multiple \channel{}s meet: the agreement group. Specifically, there are two types of messages~(i.e.,~new requests received through request channels and \textsc{Execute}s sent through commit channels) that have individual characteristics and are handled in different ways: (1)~With regard to incoming requests, agreement replicas represent the receiver side of request channels and therefore directly manage the positions of the channels' flow-control windows. As described in Section~\ref{sec:checkpointing}, to be able to quickly retrieve new requests an agreement replica updates the counter value of each client's latest request each time an agreement checkpoint becomes stable. (2)~With regard to outgoing \textsc{Execute}s, in contrast, agreement replicas represent the sender side of commit channels and therefore depend on the respective execution group at the other end of each channel to move the flow-control window forward. To prevent a single execution group from delaying overall progress, agreement replicas in \system do not wait until they are able to submit a newly produced \textsc{Execute} to every outgoing commit channel. Instead, having completed inserting an \textsc{Execute} for a sequence number~$s$ into $n_e-z$~commit channels an agreement replica is allowed to continue; $n_e$~is the total number of execution groups in the system and $z$ a configurable value~($0 \leq z < n_e$). To inform the execution groups of trailing commit channels, once such a request is garbage-collected a replica updates the channels' window positions to sequence number~\mbox{$s+1$}. If an affected execution replica subsequently tries to receive \textsc{Execute}s for sequence numbers of $s$ or lower, the commit channel responds with an exception~(see~Section~\ref{sec:channels}). In reaction, the execution replica starts to seek a stable execution checkpoint, querying members of both its own group and others, in order to compensate for the missed messages.

\subsection{Adaptability}
\label{sec:adaptability}

\system's modular architecture makes it possible to dynamically change the number of execution groups in the system and thereby adjust to varying workloads. With the consensus protocol being limited to the agreement group, in contrast to traditional BFT systems such a reconfiguration in \system does not require complex mechanisms or subprotocols.

\headline{Adding an Execution Group}
%
To add a new execution group~$e$ to the system, a privileged admin client first starts the replicas of the group and then submits an \msg{AddGroup}{e, \mathcal{E}} message; $\mathcal{E}$ is a set containing the identity and address of each group member. As soon as the agreement process for this message is complete, agreement replicas establish an \channel{} pair~(i.e.,~a request channel and a commit channel) to the new execution group, update the execution-replica registry to reflect the changes, and start the reception of requests and the forwarding of \textsc{Execute}s. Trying to obtain an \textsc{Execute} for sequence number 0, the new replicas will be notified by their commit channels that they have fallen behind and consequently use the mechanism of Section~\ref{sec:flow} to fetch an execution checkpoint from another group. 

\headline{Removing an Execution Group}
%
To remove an existing execution group~$e$ from the system, the administrator client submits a \msg{RemoveGroup}{e} message that, once agreed on, causes the agreement replicas to update the execution-replica registry and close their \channel{}s to the affected group.

\subsection{Handling Faulty Clients and Replicas}

Besides enabling \system's modular architecture, \channel{}s also play a crucial role when it comes to limiting the impact faulty clients and replicas can have on the system. In this context, especially one \channel property is of major importance: the fact that a channel only delivers a message after $f+1$~senders submitted it and the channel therefore has proof that at least one correct sender vouches for the message's validity~(see Section~\ref{sec:channels}). If, for example, a faulty client either sends conflicting requests to an execution group or the same request to fewer than \mbox{$f_e+1$}~execution replicas, the request channel of the affected execution group prevents the message's delivery to the agreement group. Note that in such case the effects of the faulty client are strictly limited to the subchannel of this client, which will not deliver a request if fewer than $f_e+1$~execution replicas insert the same message. As execution replicas use a dedicated request subchannel for each client, the subchannels of correct clients remain unaffected.

If faulty execution replicas collaborate with a faulty client, different agreement replicas may receive different values for this client's requests.
For example, a faulty client might submit a different request~$R_1$, $R_2$, \dots, $R_{f_e+1}$ to each of the $f_e+1$~correct execution replicas of one group and provide all requests to the $f_e$~faulty execution replicas of that group.
Depending on which of the request versions the faulty execution replicas transmit to which agreement replica, in such a situation it is possible that some agreement replicas obtain an $f_e+1$~quorum for request~$R_1$ while others receive $f_e+1$~matching messages for request~$R_2$ and so on.
Again, the effects are limited to the faulty client's subchannel, requests of correct clients can proceed as usual.
This scenario is not specific to \system, but in a similar way can also occur in traditional BFT systems~\cite{castro99practical,yin03separating,veronese10ebawa,sousa15separating,eischer19scalable}, in which clients directly submit their possibly conflicting requests to the replicas performing the agreement.
Consequently, all BFT protocols that tolerate faulty clients already comprise mechanisms to handle this scenario.
This is usually combined with only executing client requests with a counter value which is higher than the highest value processed so far for that client,\linebreak which ensures that old or duplicate requests are skipped.

Besides tolerating faulty clients, agreement protocols in general also provide means that allow correct follower replicas to elect a new leader if the current leader is faulty and, for example, fails to start the consensus process for a new client request within a given timeout~\cite{castro99practical,yin03separating,veronese10ebawa,sousa15separating,eischer19scalable}. To be able to monitor the leader, follower replicas must obtain information about incoming requests. In \system, this is ensured by the fact that request channels only garbage-collect a request from a correct client if the latter has successfully obtained a valid reply. A request for which this is not the case will be uploaded to all correct members of the client's execution group and through this group's request channel eventually reach all correct follower agreement replica, thereby enabling followers to hold the leader accountable.

In addition, faulty agreement replicas cannot forward manipulated messages via the commit channel.
As the consensus process ensures that all correct agreement replicas deliver the same total order of requests, eventually $f_a+1$~correct agreement replicas will send matching messages  enabling the execution groups to receive the correctly ordered requests.
 In contrast, the delivery of faulty requests sent by the faulty agreement replicas is prevented by the \channel.


%% file: sections/channel.tex
\section{\channel Implementations}
\label{sec:implementations}

In this section, we present two different variants to implement inter-regional message channels, focusing on simplicity~(\channela) and efficiency~(\channelb), respectively. Additional variants are possible, as discussed in Section~\ref{sec:related}.

\headline{\channel with Receiver-side Collection~(\channela)}
%
The receiver endpoint of an \channel only delivers a message~$m$ for a specific subchannel~$sc$ and position~$p$ if at least $f_s+1$~senders previously instructed the channel to transmit a message with identical content for the same subchannel position~(see Section~\ref{sec:channels}). As illustrated in Figure~\ref{fig:implementations-a}, the \channela solves this problem by each sender endpoint~$S_x$ directly forwarding a \smsg{Send}{m, sc, p}{S_x, \mathcal{X}} message and thereby enabling each receiver endpoint to individually collect $f_s+1$ matching messages. To allow receivers to verify the origin and integrity of a \textsc{Send}, a sender signs messages with its private key~$\mathcal{X}$.
When a receiver requests a subchannel's flow-control window to be shifted, its receiver endpoint~$R_y$ submits a signed \smsg{Move}{sc, p}{R_y, \mathcal{Y}} message to all sender endpoints. For each receiver and subchannel, a sender endpoint stores the \textsc{Move} message with the highest position~$p$ and sets the subchannel's window start to the $f_r+1$~highest position requested by any receiver~(see Section~\ref{sec:channels}).
To request a shift of a subchannel's flow-control window, sender endpoints also send \textsc{Move} messages which the receivers process analogously.


\begin{figure}
	\vspace{-2mm}
	\subfloat[\channela\label{fig:implementations-a}] {
		\includegraphics[page=1]{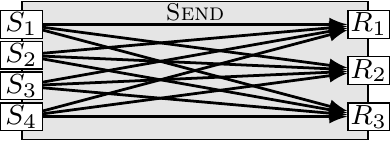}
	}
	\hfill
 	\subfloat[\channelb\label{fig:implementations-b}] {
		\includegraphics[page=2]{figures/implementations.pdf}
	}
	\caption{Overview of two possible \channel implementations.}
	\label{fig:implementations}
\end{figure}

\headline{\channel with Sender-side Collection~(\channelb)}
%
\channelb{}s minimize the number of messages transferred across wide-area links by applying the concept of \emph{collectors}~\cite{gueta19sbft}. That is, sender endpoints in \channelb{}s do not submit their \textsc{Send}s to the receiver side but, as indicated in Figure~\ref{fig:implementations-b}, instead exchange signed hashes of them within the sender group. Each sender endpoint serves as a collector, which means that the endpoint assembles a vector~$\vec{v}$ of $f_s+1$~correct signatures from different senders for the same \textsc{Send} message content~$sm$. Having obtained this vector, a collector~$S_x$ sends it in a signed \smsg{Certificate}{sm, \vec{v}}{S_x, \mathcal{X}} message to one or more receiver endpoints. On reception, a receiver verifies the validity of the \textsc{Certificate} by checking both the signatures of the message and the $f_s+1$~signatures contained in the vector~$\vec{v}$. If all of these signatures are correct and match the \textsc{Send} message content~$sm$, the endpoint has proof that $sm$ is valid as it was sent by at least one correct replica and delivers the associated message to its receiver on request.

\channelb receiver endpoints individually select the sender endpoint serving as their current collector and announce these decisions attached to their \textsc{Move}s. As a protection against faulty collectors, all sender endpoints periodically transmit \smsg{Progress}{\vec{p}}{S_x, \mathcal{X}} messages directly to receiver endpoints in which they include a vector~$\vec{p}$ with the highest position of each subchannel for which they have a \textsc{Certificate}. If at least $f_s+1$~sender endpoints claim to have reached a certain position but a receiver's collector fails to provide a corresponding and valid \textsc{Certificate} within a configurable amount of time, the endpoint switches to a different collector.


%% file: sections/evaluation.tex
\section{Evaluation}

In this section, we experimentally evaluate \system in comparison to existing approaches for BFT wide-area replication.
 
\headline{Environment}
%
To compare different techniques, we implemented a Java-based prototype that can be configured to reflect three different system architectures~(cf.~Section~\ref{sec:background-approaches}): (1)~\textbf{\bft} represents the traditional approach of distributing a single set of replicas across different geographic locations. It relies on PBFT~\cite{castro99practical} as agreement protocol and uses HMAC-SHA-256 as MACs to authenticate the messages exchanged between replicas. (2)~\textbf{\hft} employs a hierarchical system architecture running the two-level Steward protocol~\cite{amir10steward} to coordinate multiple sites that each host a dedicated cluster of replicas. Steward requires threshold cryptography for which \hft uses the scheme proposed by Shoup~\cite{shoup00practical} based on 1024-bit RSA signatures. (3)~\textbf{\system} represents our system architecture proposed in this paper. In this evaluation, \system's agreement group runs PBFT for consensus and its \channel{}s protect their messages with 1024-bit RSA signatures.


\begin{figure}[b!]
	\includegraphics{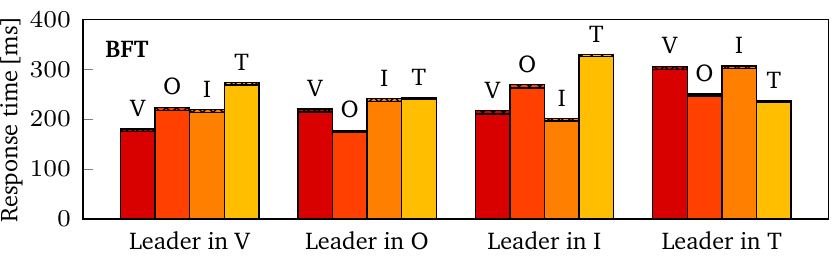}%
	\vspace{-1mm}\\
	\includegraphics{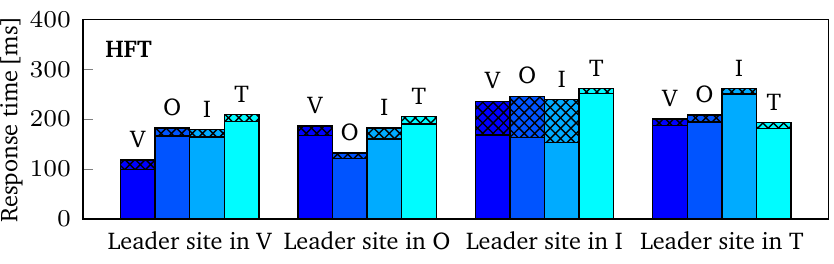}%
	\vspace{-.5mm}\\
	\includegraphics{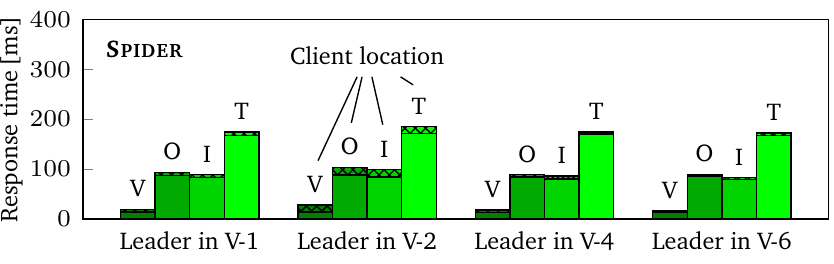}%
	\caption{50th~(\raisebox{-.3mm}{\protect\tikz \protect\node[draw, minimum width=7, minimum height=7] {};}) and 90th~(\raisebox{-.3mm}{\protect\tikz \protect\node[draw, minimum width=7, minimum height=7, postaction={pattern=crosshatch, pattern color=black}] {};}) percentiles of write latencies for different client and leader locations including Virginia~(V), Oregon~(O), Ireland~(I), and Tokyo~(T).}
	\label{fig:eval-leader-all}
\end{figure}


To conduct our experiments in an actual wide-area environment, we start virtual machines~(t3.small, 2\,VCPUs, 2\,GB\,RAM, Ubuntu\,18.04.4\,LTS, OpenJDK\,11) in 4~Amazon EC2 regions across the globe~(Virginia, Oregon, Ireland, and Tokyo). In each of these regions, we deploy 50~clients that issue 100 writes/reads per second (200~bytes) to a key-value store provided by our systems under test; client messages carry 1024-bit RSA signatures. Given this client setting, our architectures demand the following replica placement for~$f=1$: For \bft, 1~replica is hosted in each of the 4~regions. \hft expects a cluster of 4~replicas in each region, which is used as contact cluster for local clients. For \system, we deploy 1~execution group (3~replicas) per region, distributed across different availability zones. In addition, we start \system's \linebreak 4~agreement replicas in separate Virginia availability zones.

\headline{Writes}
%
In our first experiment, we examine the latency of writes issued by clients at different sites. Based on the results presented in Figure~\ref{fig:eval-leader-all}, we make three important observations: (1)~In all evaluated architectures the response times to a major degree depend on a client's geographic location. For \bft and \hft, clients in Virginia for example benefit from the fact that their local replica~(cluster) experiences comparably short round-trip times when communicating with its counterparts in Oregon and Ireland. In particular, this results in low latency when the Virginia replica (cluster) acts as leader of the wide-area consensus protocol and is able to reach a quorum together with these two other sites. In \system, clients in Virginia also observe low write latency, but for a different reason. Here, the fact that the agreement group resides in the same region as the clients' local execution group allows clients in Virginia to achieve response times of as low as 13~milliseconds. (2)~For each client location, \system \mbox{provides} significantly lower latency than \bft~(up to 95\,\%) and \hft~(up~to~94\,\%). This is a direct consequence of the fact that in contrast to the other two system architectures \system does not execute a full-fledged replication protocol over wide-area links. Instead, a write request only has to wait for two wide-area hops: from a client's local execution group to the agreement group and back. The distribution of the ordered write request to other execution groups is handled by the agreement group and thus does not require execution groups to explicitly wait for each other. That is, when an execution replica in \system receives an \textsc{Execute} for a write from the agreement group, the replica can immediately process the operation and return a reply to the client. (3)~The response times of \bft and \hft vary considerably depending on the position of the current leader of the wide-area consensus protocol. \hft clients in Ireland, for example, experience a 53\,\% higher latency when the leader is positioned in Tokyo compared to when the leader role is assigned to Virginia. In contrast, the specific location of the agreement-group leader in \system only has a negligible effect on overall response times due \linebreak to all agreement replicas residing in the same region, resulting in stable response times even across leader changes.


\begin{figure}
	\vspace{-2mm}
	\subfloat[Strongly consistent reads\label{fig:eval-readonly-strong}] {
		\includegraphics[clip, trim=0 1mm 0 0]{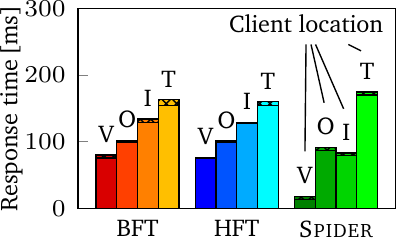}
	}
	\subfloat[Weakly consistent reads\label{fig:eval-readonly-weak}] {
		\includegraphics[clip, trim=0 1mm 0 0]{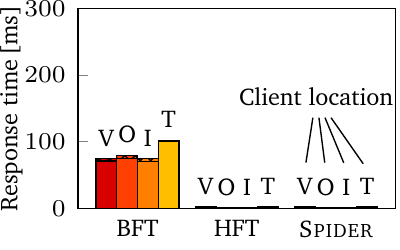}
	}
	\caption{50th~(\raisebox{-.3mm}{\protect\tikz \protect\node[draw, minimum width=7, minimum height=7] {};}) and 90th~(\raisebox{-.3mm}{\protect\tikz \protect\node[draw, minimum width=7, minimum height=7, postaction={pattern=crosshatch, pattern color=black}] {};}) percentiles of read latencies.}
	\label{fig:eval-readonly}
\end{figure}

\headline{Reads}
%
In our second experiment, we compare the evaluated architectures regarding the performance of their individual (fast-)paths for read operations with different consistency guarantees. As the results in Figure~\ref{fig:eval-readonly} show, response times of strongly consistent reads in \system display a similar pattern as writes due to following the same path through the system. For clients in Tokyo, this leads to slightly higher response times compared with \bft and \hft, which in this case benefit from directly querying replicas without intermediaries in between. For all other client locations, \system's approach, which only requires waiting for one wide-area round trip from a client's execution group to the agreement group and back, enables lower latency than provided by \bft and \hft. With regard to weakly consistent reads, both \hft and \system achieve response times of 2~milliseconds or less, as these operations can be entirely handled by replicas in a client's vicinity and therefore do not require wide-area communication as in \bft.

\headline{Modularity Impact}
%
In our third experiment, we quantify the impact of our decision to design \system as a modular architecture that separates agreement from execution and consists of loosely coupled replica groups connected via \channel{}s. We create two variants of \system where (1)~the agreement group also executes requests and is the only group in the system~(\systemze) and (2)~there is only one execution group that is co-located with the agreement group in Virginia~(\systemoe). While, \systemze allows us to study \system without \channel and externalized execution, based on \systemoe we can assess the influence of an \channel without wide-area delays. Our results show that when clients access \systemze and \systemoe from different sites, response times are dominated by the wide-area communication between clients and replicas. Thus, the modularization overhead is small and adds less than 14~milliseconds~(see Figure~\ref{fig:eval-leader-spider}).


\begin{figure}[b!]
	\subfloat[Overall latency (200-byte writes)\label{fig:eval-leader-spider}] {
		\includegraphics{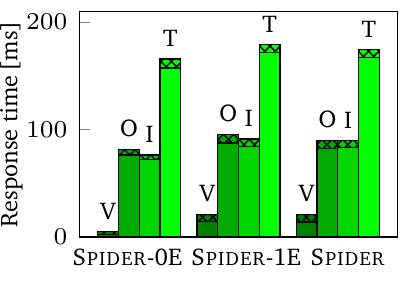}
	}
	\subfloat[Throughput\label{fig:eval-channels-lat}] {
		\includegraphics{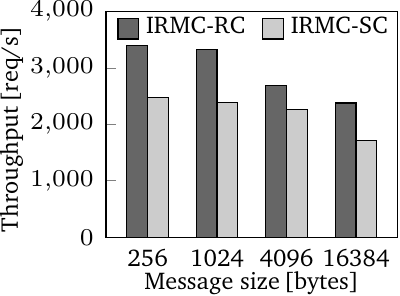}
	}\\
	\subfloat[CPU usage\label{fig:eval-channels-cpu}] {
		\includegraphics{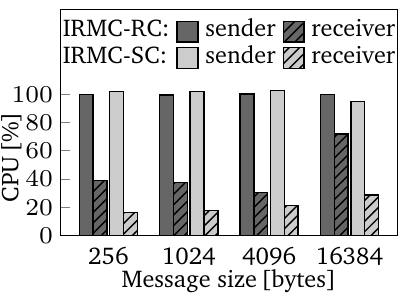}
	}
	\subfloat[Network usage\label{fig:eval-channels-net}] {
		\includegraphics{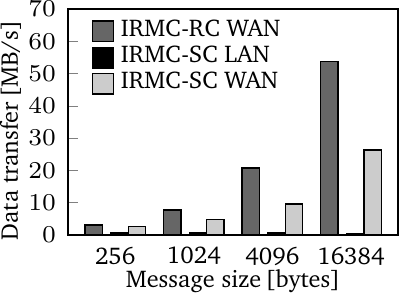}
	}
	\caption{Performance and resource usage of \channel{}s.}
	\label{fig:eval-channels}
\end{figure}

\headline{\channel{} Implementations}
%
In our fourth experiment, we evaluate the two \channel variants presented in Section~\ref{sec:implementations} by establishing a channel of each type between Virginia and Tokyo and submitting messages of different sizes. The comparison of results in Figures~\ref{fig:eval-channels-lat}--\ref{fig:eval-channels-net} confirm the two implementations to have individual characteristics. Without the need to verify signatures for \textsc{Certificate} messages, \channela sender endpoints require less CPU resources per message and therefore enable \channela{}s to achieve a higher maximum throughput. On the other hand, forwarding only one wide-area message per receiver endpoint \channelb{}s significantly reduce the amount of data transferred over long-distance links, thereby saving costs in public-cloud environments.

\headline{Adaptability}
%
In our fifth experiment, we evaluate the write and read performance new clients experience when they join the system at an additional location. For this purpose, we start with our usual setting and after 80~seconds launch 50~clients in the EC2 region Sao Paulo. Once running, the new clients in \bft and \hft issue their requests to existing replicas, while for \system they contact an additional execution group also set up in Sao Paulo. Involving more client sites than replica sites in \bft and \hft, the setting in this experiment represents a typical use-case scenario for weighted-voting approaches~(see Section~\ref{sec:background-approaches}). We therefore repeat the experiment with a fourth system~(\bftwv) that extends \bft with weighted voting and comprises a replica at each of the five client locations. As required by weighted voting, two of the five replicas are assigned higher weights in the consensus protocol. Specifically, these are the replicas in Virginia and Oregon because this weight distribution achieves the best performance in our evaluation scenario. Figure~\ref{fig:eval-newloc-avgs} presents the results of this experiment showing the average response times observed across all active client sites. To save space, we omit the results for strongly consistent reads as they show a similar picture as writes. For each system, we evaluate different leader locations, but for clarity Figure~\ref{fig:eval-newloc-avgs} only reports the results of the configurations achieving the lowest response times for each system.


\begin{figure}
	\vspace{-2mm}
	\subfloat[Writes\label{fig:eval-newloc-avg-write}] {
		\includegraphics{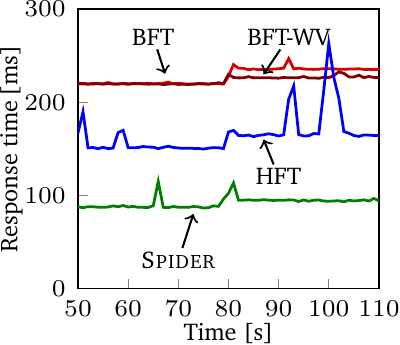}
	}
	\subfloat[Weakly consistent reads\label{fig:eval-newloc-avg-read}] {
		\includegraphics{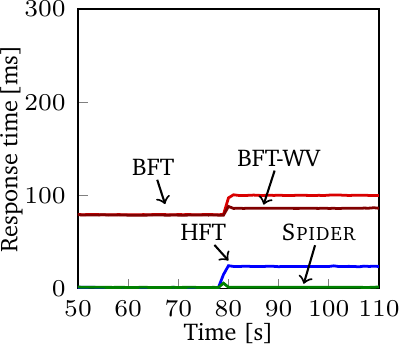}
	}
	\caption{Impact of a new client site on overall latency.}
	\label{fig:eval-newloc-avgs}
\end{figure}


Figure~\ref{fig:eval-newloc-avg-write} shows that the overall write latency increases for all evaluated architectures once the clients in Sao Paulo join the system. This is a consequence of the fact that due to its geographic location EC2's Sao Paulo region has comparably high transmission times to other cloud regions. Clients in Sao Paulo therefore observe response times between about 124~milliseconds (\system) and about 298~milliseconds (\bft), which alone causes the measurable jumps in the overall write latency averages; the response times for clients in other regions remain unaffected.
Interestingly, \bft and \bftwv achieve similar write performance throughout the experiment and thereby confirm that weighted voting does not automatically improve response times. This is only true when the additional replica is located at a site that is better connected than the existing ones and therefore enables the wide-area consensus protocol to reach faster quorums. In the setting evaluated here, \bft's typical consensus quorum is based on the votes of the replicas in Virginia, Oregon, and Ireland and therefore already provides better performance than any combination that includes the replica in Sao Paulo.

As shown in Figure~\ref{fig:eval-newloc-avg-read}, of the evaluated architectures \system is the only one that allows the new clients in Sao Paulo to perform weakly consistent reads with low latency. While all other systems require the clients in Sao Paulo to read from at least one remote replica and consequently experience overall read-latency increases of up to 23~milliseconds, \system makes it possible to introduce an execution group in the new region to efficiently handle the reads of local clients.


\begin{figure}
	\includegraphics{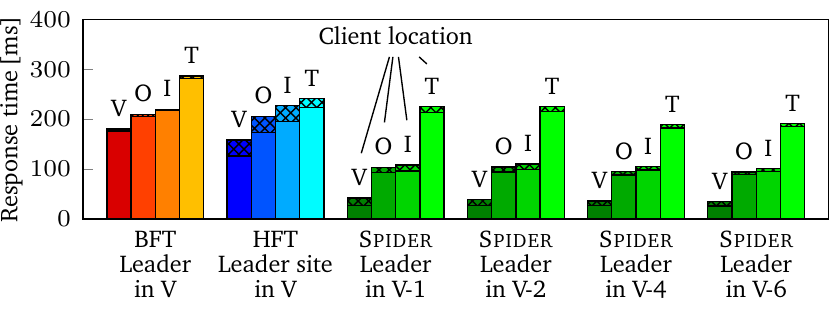}%
	\vspace{-0.5mm}
	\caption{50th~(\raisebox{-.3mm}{\protect\tikz \protect\node[draw, minimum width=7, minimum height=7] {};}) and 90th~(\raisebox{-.3mm}{\protect\tikz \protect\node[draw, minimum width=7, minimum height=7, postaction={pattern=crosshatch, pattern color=black}] {};}) percentiles of write latencies for different client sites when tolerating $f=2$ faults.}
	\label{fig:eval-leader-all-f2}
\end{figure}

\headline{Tolerating Two Faults}
%
In our final experiment, we examine write latencies for settings that are configured to tolerate $f=2$ faults in each agreement and execution group.
We place the additional replicas into nearby EC2 regions (Ohio, California, London, Seoul) to make use of further fault domains.
The results in Figure~\ref{fig:eval-leader-all-f2} show that due to increased communication latency within groups both \hft and \system see a moderate increase of response times by up to 46 milliseconds compared with the $f=1$~setting, with \system still providing significantly lower latency than \bft and \hft.

%% file: sections/related.tex
\section{Related Work}
\label{sec:related}


\headline{Adaptive BFT Replication}
%
\system is not the first work to argue that it is crucial to enable BFT systems to dynamically adapt to changing conditions. Abstract~\cite{aublin15next} makes it possible to substitute the consensus protocol of a BFT system at runtime, for example, switching to a more robust algorithm once a replica failure has been suspected or detected. CheapBFT~\cite{kapitza12cheapbft} and ReBFT~\cite{distler16resource} follow a similar idea by comprising two different agreement protocols~(one for the normal case and one for fault handling) of which only one is active at a time. In contrast, the reconfiguration mechanism developed by Carvalho et al.~\cite{carvalho18dynamic} for BFT-SMaRt~\cite{bessani14state} temporarily runs two consensus algorithms in parallel to achieve a more efficient switch. As a result of \system's modularity, integrating support for the dynamic substitution of the agreement protocol is feasible and the use of customized protocols designed for high performance~\cite{martin06fast,behl15consensus} or strong resilience~\cite{amir10prime,aublin13rbft} \linebreak would not require modifications to execution groups.

Other works allow BFT systems to dynamically change specific protocol properties at runtime. Depending on the current workload, de S\'{a} et al.~\cite{desa13adaptive}, for example, vary the parameters deciding how many requests are batched together and ordered within a single consensus instance. Berger~et~al.~\cite{berger19resilient} rely on a weighted voting scheme~\cite{sousa15separating} and by changing weights adjust the individual impact a replica has on the outcome of the agreement process. While adapting the batch size can be a measure to improve the performance of \system's agreement group, the use of a weighted voting scheme in general is only effective if (1)~a system contains more than the minimum number of agreement replicas and (2)~agreement replicas are located in different geographic regions; both of these points do not apply to \system.

\headline{Communication Between Replica Groups}
%
Amir et al. proposed BLinks~\cite{amir07customizable} as a means to send the totally ordered outputs of one replicated state machine to another replicated state machine that uses them as inputs. Unfortunately, the requirement of a channel-wide total order prevents \system from relying on BLinks as execution replicas do not necessarily have to use the same order when submitting new requests to the agreement group via their request channels. \channel{}s, on the other hand, do not have this restriction and furthermore comprise a built-in flow-control mechanism that represents the basis of \system's global flow control. However, transmitting only a single message between one dedicated sender and one dedicated receiver, BLinks may be used as a template for an \channel{} implementation \linebreak that involves even fewer wide-area messages than \channelb.

\headline{Partitioned Agreement Groups}
%
GeoBFT~\cite{gupta20resilientdb} makes use of replica groups located in different regions, which each run a full agreement protocol.
In each protocol round every group orders a request yielding a request certificate, which is shared with all other groups.
Afterwards the requests are merged into a single total order and are executed.
This requires all groups to distribute a certificate in every round, even if it just contains a placeholder request, and thus all groups must work at the same time to make progress.
In \system this requirement only applies to the agreement group whereas a limited number of slow execution groups can be skipped.
Sharing a request ordering certificate in GeoBFT works by having the leader replica forward it to $f+1$~replicas of each group, which then further forward the certificate within their group.
This request distribution scheme represents a middle ground between BLinks and \channelb{}s.
Unlike \channel{}s it is coupled with the agreement protocol and has to remotely trigger a view-change to replace a leader replica which does not complete the request distribution in a timely manner.

\headline{Efficient Client Communication}
%
In most BFT systems, clients need to receive replies from different replicas in order to prove a result correct~\cite{castro99practical}, which in geo-replicated settings can significantly increase the number of messages exchanged over wide-area links. SBFT~\cite{gueta19sbft} addresses this problem by adding a protocol phase that aggregates request acknowledgements of multiple replicas into a single message to the client. In Troxy~\cite{li18troxy}, a client also has to wait for a single reply only, because the reply voter is hosted inside a trusted domain at the server side and forwards its decisions to the client through a secure channel. In \system, clients are typically located in the same region as an execution group allowing for communication over short-distance links. For scenarios in which this is not the case, it would be possible to extend \system to use one of the approaches discussed above.

\headline{Leader Selection in Geo-replicated Systems}
%
Multiple authors have underlined the impact that the leader-replica location has on response times, independent of the fault model, and presented solutions to select the leader in a way that minimizes overall latency~\mbox{\cite{sousa15separating,liu17leader,eischer18latency}}. Other agreement-based systems do not need to determine a fixed leader as they continuously rotate the leader role among replicas~\cite{veronese09spin,mao09towards,veronese10ebawa,mao08mencius,milosevic13bounded}. As our experiments show, with agreement replicas residing in different availability zones of the same cloud region, the specific location of the consensus leader in \system only has a negligible effect on response times. Consequently, \system achieves low and stable latency without requiring means to dynamically select or rotate the leader.

\headline{Crash-tolerant Wide-Area Replication}
%
Several works addressed the efficiency of geo-replication in systems that unlike \system solely tolerate crashes, not Byzantine failures. In Pileus~\cite{terry13consistency}, for example, writes are only handled by a subset of replicas that first order and execute them, and then bring all other replicas up to date by transferring state changes. P-Store~\cite{schiper10pstore} improves efficiency in wide-area environments by performing partial replication, thereby freeing a site from the need to receive and process all updates. Clock-RSM~\cite{du14clock} establishes a total order on requests by exploiting the timestamps of physical clocks and without requiring a dedicated leader replica. EPaxos~\cite{moraru13there} in contrast does not rely on a total request order, but only orders those requests that interfere with each other due to accessing the same state parts.


%% file: sections/conclusion.tex
\section{Conclusion}

The cloud-based \system system architecture models a BFT system as a collection of loosely coupled replica groups that can be flexibly distributed in geo-replicated environments. In contrast to existing approaches, \system does not require the execution of complex multi-phase protocols over wide-area links, but instead performs essential tasks such as consensus, leader election, and checkpointing across replicas residing in the same region. Our experiments show that this approach enables \system to achieve low and stable response times.


%% file: sections/proof.tex

\lstset{
	basicstyle=\footnotesize\lsttt,
	emphstyle=\lstbtt,
	commentstyle=\commentsize\textit,
	tabsize=2,
	numberstyle=\scriptsize,
	numbersep=2.2mm,
	xleftmargin=5mm,
	frame=none,
	columns=fullflexible,
	numberblanklines=true,
	emptylines=2,
	breaklines=true,
	breakatwhitespace=false,
	escapechar=\%,
	mathescape,
	morecomment=[l]{//},
	morecomment=[s]{/*}{*/},
	morestring=[b]",
	aboveskip=-3mm,
	emph={%
		repeat, until, write, broadcast, to, sleep, for, until, return, on, receive, from, if, send, unwrap, main, loop, while, else, parallel, each, deliver, execute, order, move_window, valid_mac, fetch_cp, stable_cp, gen_cp, gc, case, periodic%
	}
}
\newtheorem*{remark}{Remark}

\section{Safety and Liveness Proof for \system}
In the following, we first provide a detailed description of the individual components of \system, along with the assumptions and definitions used for proving the correctness and liveness properties of \system.
Afterwards, we present the proof itself and conclude with pseudocode for both IRMC implementation variants (IRMC-RC and IRMC-SC).

\subsection{Fault Assumptions}
We assume that each execution group consists of $2f_e+1$~replicas and that there are up to $f_e$~faulty execution replicas per execution group.
The agreement group has $3f_a+1$~replicas of which up to $f_a$~agreement replicas may be faulty.
All faults are assumed to be Byzantine.

We assume a partially synchronous network with periods of synchrony which are long enough to allow the protocol to make progress~\cite{dwork88consensus}.

\def\agr{agreement replica\xspace}
\def\agrs{agreement replicas\xspace}
\def\exec{execution replica\xspace}
\def\execs{execution replicas\xspace}
\def\ag{agreement black-box\xspace}

\subsection{Cryptographic Primitives and Assumptions}
The pseudocode uses the following cryptographic primitives:
\begin{itemize}
\item sign($m$): Digitally sign message $m$ (e.g., using RSA).
\item valid\_sig\textsubscript{$\mathcal{E}$}($m$): Verify that the signature for message~$m$ is valid and that the signer is part of group $\mathcal{E}$.
\item mac\textsubscript{$a,e$}(m): Add a single MAC~(message authentication code) such that replica~$a$ authenticates message~$m$ towards replica~$e$~\cite{tsudik92message}. This primitive, for example, may be implemented using HMAC-SHA-256.
\item mac\textsubscript{$a,\mathcal{E}$}(m): Add a MAC vector such that replica~$a$ authenticates message~$m$ to a replica group $\mathcal{E}$~\cite{castro99practical}. It consists of a MAC for each replica in group~$\mathcal{E}$.
\item valid\_mac\textsubscript{$a,e$}(m) and valid\_mac\textsubscript{$a,\mathcal{E}$}(m) are used to verify these MACs.
\item unwrap\_mac(m): Strips the added MAC from message~$m$ and returns the original message.
\item h($m$): Calculate a cryptographically secure hash digest of message~$m$, for example using SHA-256.
\end{itemize}

\noindent{}We make the standard assumptions regarding cryptographic functions.
We assume them  to be secure, that is a malicious replica cannot forge signatures\,/\,MACs of other replicas nor can it create a message $m' \neq m$ with hash~$h(m) = h(m')$.

\subsection{Consistency Guarantees}
\system provides linearizability for write requests.
Read requests with strong consistency are treated similarly, but only the designated execution group gets the full request, whereas all other groups just receive the client id and counter.

Weakly consistent reads provide one-copy serializability.

Section~\ref{sec:consistency-guarantees} contains the relevant proofs and definitions of the consistency guarantees.

\subsection{Definitions}
We first describe the properties provided by \system before describing the required assumptions for the agreement black-box and the checkpoint component.
\subsubsection{Properties of \system}
The definitions of E-Safety and E-Validity follow the lines of those used for Steward~\cite{amir10steward}.
E-Safety~II and E-Liveness are adapted from PBFT~\cite{castro99practical}.
E-Validity~II captures the usual at-most-once guarantee.
\begin{definition}[E-Safety]
	\label{def:e-safety}
	If two correct servers execute the i\textsuperscript{th} write, then these writes are identical.
\end{definition}
\begin{definition}[E-Safety II]
	\label{def:e-safety2}
	The system provides linearizability regarding requests from correct clients.
\end{definition}
\begin{definition}[E-Validity]
	\label{def:e-validity}
	Only a correctly authenticated write request from a client may be executed.
\end{definition}
\begin{definition}[E-Validity II]
\label{def:e-validity2}
A write request may be executed at most once.
\end{definition}
\begin{definition}[E-Liveness]
	\label{def:e-liveness}
	A correct client will eventually receive a reply to its request.
\end{definition}

\subsubsection{Agreement Black-Box}
\begin{figure}
	\vspace*{3mm}%
	\begin{lstlisting}
interface %\lstbtt{Agreement}% {
	// Request ordering of message m
	void %\lstbtt{order}%(%\textsc{Message}% m);
	// Must deliver request in order without gaps
	// Blocking callback, that is the agreement can only deliver the next message after the previous deliver call has completed
	// Delays in deliver may cause timeouts in the agreement black-box to expire
	callback %\lstbtt{deliver}%(%\textsc{SeqNr}% s, %\textsc{Message}% m);
	
	// Forget everything before (%$<$%) sequence number s
	// After this call no sequence number %$<$% s must be  delivered
	void %\lstbtt{gc}%(%\textsc{SeqNr}% s);
}
	\end{lstlisting}
	\caption{Agreement black-box interface (pseudo code)}
	\label{def:ag-black-box}
\end{figure}
We assume the agreement to be a black-box with the interface shown in Figure~\ref{def:ag-black-box} and the following properties.
The comments at the interface methods detail their expected behavior. We assume that the first delivered sequence number is 1.
\begin{definition}[A-Safety]
	\label{def:a-safety}
	If two correct agreement replicas deliver a message for sequence number~$s$, then these messages are identical.
\end{definition}
\begin{definition}[A-Liveness]
	\label{def:a-liveness}
	If $2f+1$~correct replicas receive a message $m$ for ordering, then eventually $f+1$~correct replicas will deliver message $m$ and all preceding messages.
\end{definition}
\begin{definition}[A-Validity]
	\label{def:a-validity}
	A correct agreement replica will only deliver correctly authenticated client requests.
\end{definition}
\begin{definition}[A-Order]
	\label{def:a-order}
	A correct agreement replica will deliver a message for sequence number~$s$ only after all preceding sequence numbers were delivered or garbage collected.
\end{definition}
\noindent{}These requirements are for example fulfilled by PBFT~\cite{castro99practical}.


\subsubsection{Checkpoint Component}
\begin{figure}
	\vspace*{3mm}%
	\begin{lstlisting}
interface %\lstbtt{Checkpoint}% {
	// Create and distribute own checkpoint message
	// By default only checkpoint components within a single group communicate with each other (i.e., checkpoints%\,%are%\,%group%\,%specific)
	void %\lstbtt{gen\_cp}%(%\textsc{SeqNr} $s$%, %\textsc{State} $st$%);
	// Sequence numbers of delivered checkpoints must increase monotonically
	// Older checkpoints must be skipped, if a newer checkpoint has already been delivered
	callback %\lstbtt{stable\_cp}%(%\textsc{SeqNr} $s$%, %\textsc{State} $st$%);
	// Actively fetch requested checkpoint
	void %\lstbtt{fetch\_cp}%(%\textsc{SeqNr} $s$%);
}
	\end{lstlisting}
	\caption{Checkpoint-component interface (pseudo code)}
	\label{def:cp-interface}
\end{figure}

We assume that each replica has a checkpoint component with the interface from Figure~\ref{def:cp-interface} and the following properties.
The comments at the interface methods detail their expected behavior.

\begin{definition}[Stable checkpoint]
	A checkpoint is called \emph{stable} once a correct replica collects a certificate consisting of $f+1$~valid and matching checkpoint messages.
\end{definition}
\noindent{}Once a replica possesses a stable checkpoint it will call \texttt{stable\_cp} with the checkpoint, unless it has already delivered a checkpoint with a higher sequence number.
\begin{definition}[CP-Safety]
	\label{def:cp-safety}
	A stable checkpoint was created by at least one correct replica.
\end{definition}
\noindent{}As shown later on, all correct replicas in a group will create identical checkpoints for the same sequence number.
\begin{definition}[CP-Liveness]
	\label{def:cp-liveness}
	If one correct replica of a group delivers a checkpoint, then eventually all correct replicas of that group will deliver that checkpoint, unless a newer checkpoint was already delivered.
\end{definition}
\begin{definition}[CP-Liveness II]
\label{def:cp-liveness2}
  Once $f+1$ correct replicas create and distribute identical checkpoint messages, the checkpoint will eventually become stable, unless it is superseded by a newer one before.
\end{definition}

\noindent{}An implementation should consider the following aspects:
\begin{itemize}
	\item With an execution group size of $2f_e+1$ CP-Safety requires that each checkpoint message is authenticated using a signature.
	\item In order to provide CP-Liveness correct replicas must continuously inform\,/\,query each other about their latest stable checkpoint.
	\item A checkpoint message $\langle\textsc{Checkpoint},h,s\rangle$ for sequence number~$s$ with $h=h(st)$ only contains a hash of the checkpoint state~$st$ to keep the network overhead low.
	\item The full checkpoint state should only be transferred when necessary.
\end{itemize}

\subsubsection{Application}
We assume that the application is implemented as a deterministic state machine which can execute client requests and provide a reply to them.
In addition, the application must be able to retrieve and apply a checkpoint.
The latter functionalities are denoted as assignment \texttt{app := app'} and passing \texttt{app} to \texttt{cp.gen\_cp} in pseudo code.
\begin{definition}[RSM]
	\label{def:rsm}
	Different application instances have an identical state for sequence number $i$ when processing writes according to the same total order~\cite{schneider90implementing}.
\end{definition}

\subsection{IRMC Properties}
\begin{figure}
	\vspace*{3mm}%
	\begin{lstlisting}
/* Sender endpoint%\,%*/
interface %\lstbtt{IRMC\_Sender}% {
	// If %$p$% is too old: discard %$m$% and return immediately 
	// If %$p$% is in the current window: send %$m$% and return immediately 
	// If %$p$% is after the current window (%$p > max(IRMC_{sc}.win)$%): block/wait
	void %\lstbtt{send}%(%\textsc{Subchannel}% %$sc$%, %\textsc{Position}% %$p$%, %\textsc{Message}% %$m$%);
	// Ask receiver endpoint to move the window forward
	// The receiver endpoint will internally call move_window with the %$f_s+1$%-highest received position
	void %\lstbtt{move\_window}%(%\textsc{Subchannel}% %$sc$%, %\textsc{Position}% %$p$%);
}

/* Receiver endpoint%\,%*/
interface %\lstbtt{IRMC\_Receiver}% {
	// Blocks until (1) a message %$m$% is delivered, then returns %$m$%, or until (2) the window is ahead of %$p$%, that is %$p < min(IRMC_{sc}.win)$%, then returns %$\langle\textsc{TooOld}, s\rangle$%, with %$s$% = new window lower bound
	%\textsc{Message}% %\lstbtt{receive}%(%\textsc{Subchannel}% %$sc$%, %\textsc{Position}% %$p$%);
	// Position %$p$% must increase monotonically, calls with lower values are silently ignored
	void %\lstbtt{move\_window}%(%\textsc{Subchannel}% %$sc$%, %\textsc{Position}% %$p$%);
}
	\end{lstlisting}
	\caption{IRMC interfaces (pseudo code)}
	\label{def:irmc-interface}
\end{figure}
The sender and receiver endpoint interfaces of the IRMC are shown in Figure~\ref{def:irmc-interface}.
As before, the comments specify the expected behavior of the methods.
All sender replicas are contained in the set $R_s$ and all receiver replicas in $R_r$.
The capacity of an IRMC (subchannel) is denoted as $|IRMC|$ and is assumed to be $\geq 1$.
It is identical for all subchannels of an IRMC.
$IRMC_{sc}.win$ refers to the window of subchannel $sc$, which is initialized to start at 1.
$min(IRMC_{sc}.win)$ and $max(IRMC_{sc}.win)$ return the lower and upper limit\linebreak \mbox{(inclusive)} of the window of subchannel~$sc$, respectively.
$receive(sc, p) = m$ denotes that the receive call returned the message~$m$.
\begin{definition}[IRMC-Correctness I]
	\label{def:irmc-correctness1}
	Receive only returns a message sent by a correct sender:\\
	$receive(sc, p) = m \rightarrow$ a correct sender called $send(sc, p, m)~\wedge$ the receiver called $move\_window(sc, p')$ such that $p' \leq p <  p' + |IRMC_{sc}|$.
\end{definition}
\begin{definition}[IRMC-Correctness II]
	\label{def:irmc-correctness2}
	Moving a window requires a move request by at least one correct replica:\\
	$receive(sc, p) = \langle\textsc{TooOld}, p'\rangle$ with $p' > p \rightarrow$ a correct sender called $move\_window(sc, \hat{p})$ with $\hat{p}\geq p'~\vee$ a correct receiver called $move\_window(sc, \hat{p})$ with $\hat{p}\geq p'$.
\end{definition}
\begin{remark}
	\label{def:irmc-remark}
	Calls to \texttt{send} block if the requested position is after the upper limit of the current subchannel window.
	Calls to \texttt{receive} block if the position is in or after the subchannel window and the corresponding message was not yet received by the IRMC.
\end{remark}
\begin{definition}[IRMC-Liveness I]
	\label{def:irmc-liveness1}
	An identical message sent (\texttt{send} method call has returned) by at least $f_s+1$ correct replicas will eventually cause some message to be received by all correct receivers unless it is skipped (see also IRMC-Correctness II):\\
	If $f_s+1$~correct senders call $send(sc,p,m)$, then eventually $\forall$~correct $r \in R_s$ that call(ed) $receive(sc, p)\hspace{-.5mm}:\hspace{.2mm}receive(sc, p)\hspace{-.5mm}=\hspace{-.5mm}*$ $\vee~receive(sc, p) = \langle\textsc{TooOld}, p'\rangle$ with $p' > p$.
\end{definition}
\begin{remark}
	Due to IRMC-Correctness I the received message can only be one that was sent by at least one correct sender.
\end{remark}
\begin{definition}[IRMC-Liveness II]
	\label{def:irmc-liveness2}
	Send calls return once the position is below the subchannel window's upper bound:\\
	If $f_r+1$~correct receivers $r \in R_r$ call $move\_window(sc, p_r)$, then eventually all $send(sc, p', m)$ calls will have returned on all correct sender replicas where $p' < \tilde{p} + |IRMC_{sc}|$ and $\tilde{p} = f+1$-largest $p_r$.
\end{definition}
\begin{definition}[IRMC-Liveness III]
	\label{def:irmc-liveness3}
	Receiver endpoints will move the window at least as far as the $f_s+1$-highest \texttt{move\_win} request by a sender replica:\\
	If $f_s+1$~correct senders call $move\_window(sc, p_s)$, then eventually all correct receiver endpoints will have (internally) called $move\_window(sc, p)$ with $p$ such that largest $p_s \geq p \geq f+1$-largest $ p_s$.
\end{definition}
\begin{remark}
Note that if a receiver endpoint has already moved a subchannel window to a higher position than $p$, then the call to \texttt{move\_win} has no effect.
\end{remark}

\subsection{\system Pseudo Code}
\lstset{
	numbers=left,
}
\begin{figure}
	\vspace*{3mm}%
\begin{lstlisting}
%$t_c$% := 1%\hfill%// Client request counter
%$rep$% := %$\varnothing$%%\hfill%// Reply for last request
%$g$% := %$\{\}$\hfill%// Collected replies
%$\mathcal{E}$% := nearest execution group with %$|\mathcal{E}| = 2f_e+1$%%\hrule%
write(%\textsc{Write}% %$w$%):%\label{code:client-write}%
	// Authenticate request
	%$m$% := %$mac_{c,\mathcal{E}}$%(%$sign_c$%(%$\langle\textsc{Write}, w, c, t_c \rangle$%))%\label{code:client-sign}%
	%$rep$% := %$\varnothing$%
	%$g$% := %$\{\}$%
	// Repeat sending until reply was received
	repeat until %$rep \neq \varnothing$%:%\label{code:client-loop}%
		broadcast %$m$% to %$\mathcal{E}$%
		sleep for %$t_{retry}$% %$\vee$% until %$rep \neq \varnothing$%
	%$t_c$% := %$t_c+1$%%\label{code:client-counter}%
	return %$rep$%

on receive(%$m$% = %$\langle\textsc{Reply}, u, t_c' \rangle$% from %$e \in \mathcal{E}$%):
	// Only process correctly authenticated replies
	// Each replica may only send one
	if %$valid\_mac_{e,c}$%(%$m$%)%$\wedge t_c' = t_c \wedge \langle\textsc{Reply},*,*\rangle$% from %$e\notin g$%:%\label{code:client-precheck}%
		%$g$% := %$g \cup \{m\}$%
		// Return reply after receiving %$f_e\hspace{-.3mm}+\hspace{-.3mm}1$% replies with matching %$t_c$% and %$u$%
		if %$\exists u: | \{v|v=\langle\textsc{Reply}, u, t_c\rangle\in g\}| \geq f_e+1$%:%\label{code:client-check}%
			%$rep$% := %$u$%
\end{lstlisting}
\caption{Client $c$ (pseudo code)}
\label{def:pseudo-client}
\end{figure}

\begin{figure}
	\vspace*{3mm}%
	\begin{lstlisting}
%$s_n$% := 0 %\hfill%// Sequence number for last executed request
%$t[c]$% := 0 %\hfill%// Counter of latest forwarded client request
%$u[c]$% := %$\varnothing$%%\hfill%// Reply cache %$\langle\textsc{Reply}, u_c, t_c\rangle$%
app = application, cp = checkpoint component
%$\mathcal{E}$% := execution group
%$r_\mathcal{E}$% = request IRMC sender, %$\forall c: |r_{\mathcal{E},c}| = 2$% %\hfill%// Capacity = 2
%$c_\mathcal{E}$% = commit IRMC receiver, %$|c_{\mathcal{E},0}| \geq k_e$% %\hfill%// Capacity %$\geq k_e$%%\hrule%
on receive(%$m$% = %$\langle\textsc{Write},w,c,t_c\rangle$% from %$c$%):%\label{code:exec-recv-write}%
	// Ignore invalid requests
	if %$!valid\_mac_{c,\mathcal{E}}$%(%$m$%): return%\label{code:exec-check-mac}%
	if %$t_c \leq t[c]$%:
		// Check if a reply is available for the request
		if %$u[c]=\langle\textsc{Reply}, *, t_c'\rangle \wedge t_c' = t_c$%:
			send %$mac_{e,c}$%(%$u[c]$%) to %$c$%%\label{code:exec-resend-reply}%
		return // Silent return on retry with no result yet
	if %$!valid\_sig_{c}$%(%$unwrap\_mac$%(%$m$%)): return%\label{code:exec-check-sig}%
	// Execution replicas must be able to forward a request once
	// This also applies for the latest client request if an execution replica already has a reply
	%$t[c]$% := %$t_c$%%\label{code:exec-counter}%
	// Notify agreement of new request
	%$r_\mathcal{E}$%.move_window(%$c, t_c$%)%\label{code:exec-move-win}%
	%$r_\mathcal{E}$%.send(%$c, t_c, \langle\textsc{Request},unwrap\_mac(m),\mathcal{E}\rangle$%)%\label{code:exec-send}%

main loop:%\label{code:exec-main}%
	while true:
		%$m$% := %$c_\mathcal{E}$%.receive(%$0, s_n+1$%)%\label{code:exec-receive}%
		if %$m = \langle \textsc{TooOld}, s'\rangle$%:
			// Executor missed some requests %$\rightarrow$% fetch checkpoint
			cp.fetch_cp(%$s'$%) // Ask other groups if necessary
		else:
			%$m$% = %$\langle\textsc{Execute},\langle\textsc{Request},\langle\textsc{Write},\,w,\,c,\,t_c\rangle, \mathcal{E}'\rangle, s_n+1\rangle$\label{code:exec-exec-msg}%
			%$s_n$% := %$s_n + 1$\label{code:exec-inc-steps}%
			// Filter duplicate%\,%/%\,%old request
			if %$u[c]=\langle\textsc{Reply}, *, t_c'\rangle \wedge t_c' < t_c \vee u[c]=\varnothing$%:%\label{code:exec-skip}%
				%$u_c$% := app.execute(%$m$%)%\label{code:exec-execute}%
				%$u[c]$% :=  %$\langle\textsc{Reply}, u_c, t_c\rangle$% // Store reply%\label{code:exec-store-rep}%
				if %$\mathcal{E} = \mathcal{E'}$%: // Only the local execution group sends the reply to the client
					send %$mac_{e,c}$%(%$u[c]$%) to %$c$%%\label{code:exec-send-reply}%
			if %$s_n \equiv 0$ mod $k_e$%: // Periodically create a checkpoint
				cp.gen_cp(%$s_n$%, (%$u$%, app))%\label{code:exec-gen-cp}%
			
on cp.stable_cp(%$s$%, %$st$% = (%$u'$%, app')):%\label{code:exec-stable-cp}%
	// Allow garbage collection of commit IRMC
	%$c_\mathcal{E}$%.move_window(%$0, s+1$%)%\label{code:exec-cp-move}%
	if %$s \geq s_n$%:%\label{code:exec-cp-guard}%
		%$s_n$% := %$s$%%\label{code:exec-cp-update}%
		app := app'
		%$u$% := %$u'$%
	\end{lstlisting}
	\caption{Execution replica $e$ (pseudo code)}
	\label{def:pseudo-execution}
\end{figure}

\begin{figure}
	\vspace*{3mm}%
	\begin{lstlisting}
%$s_n$% := 0 %\hfill%// Last ordered sequence number
// Force agreement to periodically create a checkpoint
%$win$\,%:=%\,%[1,AG-WIN]%\hfill%//%\,%Range%\,%with%\,%[lower,%\,%upper]%\,%bound,%\,%both%\,%inclusive
AG-WIN%\,$\geq$\,$k_a$%%$\hfill$%// Size of agreement window
%$t[c]$% := 0%\hfill%// Counter values of latest agreed request; used by consensus
%$t^+[c]$% := 0%\hfill%// Counter values for next expected request
%$hist$% := last %$|c_{\mathcal{E},0}|$ \textsc{Execute}s%
ag = agreement black-box, cp = checkpoint component
for each execution group %$\mathcal{E}$%:
	$r_\mathcal{E}$ = request IRMC receiver, %$\forall c: |r_{\mathcal{E},c}|=2$%
	$c_\mathcal{E}$ = commit IRMC sender, %$|c_{\mathcal{E},0}| \geq k_e$%
%$\mathcal{A}$% := agreement group with %$|\mathcal{A}|=3f_a+1$\hrule%
parallel for each client %$c$% and execution group %$\mathcal{E}$%:
	while true:
		%$m$% := %$r_\mathcal{E}$%.receive(%$c$%, %$t^+[c]$%)%\label{code:agree-client-receive}%
		if %$m$% = %$\langle\textsc{TooOld}, s\rangle$%:
			// Client already sent a newer request
			%$t^+[c]$% := %$s$%%\label{code:agree-too-old}%
		else: // %$m$% = %$\langle\textsc{Write}, w, c, t_c\rangle$%
			// Order request and wait for the next one
			ag.order(%$m$%)
			%$t^+[c]$% := %$t^+[c] + 1$%%\label{code:agree-next-req}%

// In-order without gaps between sequence numbers, blocks agreement, blocking can cause agreement timeouts to expire
on ag.deliver(%$s$%, %$r$% = %$\langle\textsc{Request},\langle\textsc{Write},w,c,t_c\rangle,\mathcal{E}\rangle$%):%\label{code:agree-deliver}%
	// Sleep if agreement must create a new checkpoint
	sleep until %$s \leq$% %$max(win$%)%\label{code:agree-win}%
	%$x$% := %$\langle\textsc{Execute}, r, s\rangle$%
	// Update state with new request
	// Old%\,%/%\,%duplicated requests could be replaced with no-ops here
	%$t[c]$% := %$ t_c$\label{code:agree-deliver-update}%
	%$t^+[c]$% := %$\max(t_c + 1, t^+[c])$%%\label{code:agree-next-deliver}%
	%$hist$%.add(%$x$%)%\label{code:agree-hist}%
	%$s_n$% := %$s$%
	parallel for each execution group %$\mathcal{E}$%:
		%$c_\mathcal{E}$%.send(%$0, s, x$%)%\label{code:agree-deliver-send}%
	sleep until completed for %$n_e - z$% groups%\label{code:agree-wait-send}%
	// Not completed parallel calls continue in the background
	if %$s_n \equiv 0$ mod $k_a$%: // Periodically create a checkpoint
		cp.gen_cp(%$s_n$%, (%$t$%, %$hist$%))%\label{code:agree-gen-cp}%

on cp.stable_cp(%$s$%, %$st$% = (%$t'$%, %$hist'$%)):%\label{code:agree-stable-cp}%
	// Move commit window forward
	parallel for each execution group %$\mathcal{E}$%:
		%$c_\mathcal{E}$%.move_window(%$0, s-|hist'|+1$%)%\label{code:agree-cp-move}%
	ag.gc(%$s+1$%)%\label{code:agree-cp-gc}%
	if %$s > s_n$%:%\label{code:agree-cp-guard}%
		%$s_n'$% := %$s_n$%
		%$s_n$% := %$s$%%\label{code:agree-cp-update}%
		%$t$% := %$t'$%
		%$hist$% := %$hist'$%
		parallel for each execution group %$\mathcal{E}$%:
			// Add missing requests from hist to commit IRMC
			for %$x$% = %$\langle\textsc{Execute}, r, s'\rangle \in hist, s' \in [s_n'+1, s]$%:
				%$c_\mathcal{E}$%.send(%$0, s', x$%)%\label{code:agree-cp-send}%
		sleep until completed for %$n_e - z$% groups
	%$win$% := [%$s$%+1, %$s$%+AG-WIN]
	\end{lstlisting}
	\caption{Agreement replica $a$ (pseudo code)}
	\label{def:pseudo-agreement}
\end{figure}

The pseudo code for the client is shown in Figure~\ref{def:pseudo-client}, for the \exec in Figure~\ref{def:pseudo-execution} and for the agreement replica in Figure~\ref{def:pseudo-agreement}.

We assume that each method is executed atomically, unless it calls a blocking method, at which point execution may switch to other methods.
Variable definitions are written as \texttt{var := value}, whereas \texttt{=} is used for comparisons and destructuring of values, for example $x = \langle\textsc{Execute}, r, s'\rangle$ uses the value in $x$ to define $r$ and $s'$ using pattern matching.

\subsection{Proof}
The proof primarily considers write requests.
We assume for now that there is only one execution group, that is $n_e=1$ and $z=0$.
Later on, we will relax this assumption.
Strongly and weakly consistent read requests  are considered afterwards.
We write "L.~\ref{def:pseudo-client}.\ref{code:client-write}" to refer to Line~\ref{code:client-write} in Figure~\ref{def:pseudo-client}.

\subsubsection{\hspace{1mm}Agreement-Checkpoint~~Equivalence~~(CP-A- Equivalence)}
\begin{definition}[CP-A-Equivalence]
	\label{def:cp-a-equivalence}
	The state of an agreement replica ($s_n$, $t$, $hist$ and queued commit IRMCs messages) that has reached sequence number~$s$  via processing \texttt{ag.deliver}($s, r$) (L.~\ref{def:pseudo-agreement}.\ref{code:agree-deliver}) is equivalent to that of a replica that reaches sequence number~$s$ by applying a checkpoint for sequence number~$s$.
\end{definition}
\begin{proof}
	We prove this by induction.\\
	\indent{}\emph{Base case}: All correct agreement replicas initialize $s_n$, $t$, $hist$ and the commit IRMCs with identical values. There is no checkpoint for that sequence number, as no checkpoint was generated yet.\\
	\indent{}\emph{Induction step}: All correct agreement replicas pass through the same states by processing ordered requests or jump forward to one of those states via a checkpoint.
	
	As updates to the considered state parts are only made in either \texttt{ag.deliver}~(L.~\ref{def:pseudo-agreement}.\ref{code:agree-deliver}) or \texttt{cp.stable\_cp}~(L.~\ref{def:pseudo-agreement}.\ref{code:agree-stable-cp}), it suffices to show that when either of them updates $s_n$ to a certain sequence number, then the resulting replica states are equivalent.
	Note that the sequence number~$s_n$ increases monotonically as \texttt{ag.deliver} is per A-Order~\ref{def:a-order} only called for increasing sequence numbers and \texttt{cp.stable\_cp} only increases the value of $s_n$~(L.~\ref{def:pseudo-agreement}.\ref{code:agree-cp-guard}).
	
	Assume that from a common starting point, replicas reach sequence number~$s$ by processing \texttt{ag.deliver}($s, r$)~(L.~\ref{def:pseudo-agreement}.\ref{code:agree-deliver}):
	Per A-Safety~\ref{def:a-safety} and A-Order~\ref{def:a-order} all correct \agrs receive the same sequence of requests via their \texttt{ag.deliver} callback, that is $s_n$, $t$ and $hist$~(L.~\ref{def:pseudo-agreement}.\ref{code:agree-deliver-update}) evolve identically on those replicas.
	Therefore, a possible later call to \texttt{cp.gen\_cp}($s$, ($t, hist$))~(L.~\ref{def:pseudo-agreement}.\ref{code:agree-gen-cp}) for a sequence number~$s$ has identical parameters on all correct \agrs.
	
	As per CP-Safety~\ref{def:cp-safety} only checkpoints which were created by at least one correct replica  can become stable, any call of \texttt{cp.stable\_cp}($s$, ($t'$, $hist'$))~(L.~\ref{def:pseudo-agreement}.\ref{code:agree-stable-cp}) can only deliver that checkpoint for sequence number~$s$.
	Applying a checkpoint for the current or an older sequence number~$s \leq s_n$ does not change $s_n$, $t$ and $hist$~(L.~\ref{def:pseudo-agreement}.\ref{code:agree-cp-guard}).
	Applying a checkpoint for a newer sequence number~$s > s_n$ atomically sets $s_n$, $t$ and $hist$ to the state they had when the checkpoint was created~(L.~\ref{def:pseudo-agreement}.\ref{code:agree-cp-update}) and adds missing requests (i.e., those skipped by updating $s_n$) to the commit IRMCs.
	The call to \texttt{ag.gc}($s+1$), which happens atomically with the state update, ensures that \texttt{ag.deliver} will only be called for sequence numbers $\geq s+1$.
	Per A-Order~\ref{def:a-order} the next \texttt{ag.deliver} call must be for $s_n+1 = s+1$.
	
	When called for an old checkpoint ($s \leq s_n$), then \texttt{$c_\mathcal{E}$.move\_window} (L.~\ref{def:pseudo-agreement}.\ref{code:agree-cp-move}) has no effect, as a \texttt{$c_\mathcal{E}$.send} call for $s_n$ must already have been issued, such that the IRMC has queued messages at least up to sequence number~$s$.
	Therefore $max(c_{\mathcal{E},0}.win) \geq s \Leftrightarrow min(c_{\mathcal{E},0}.win) \geq s-|c_{\mathcal{E},0}|+1$ that is the window start is at least at the position requested by the \texttt{$c_\mathcal{E}$.move\_window} call, see also the remark below.
	
	For a newer checkpoint, as $|hist'| = |c_{\mathcal{E},0}|$, this together with moving the window forward from the sender-side (per IRMC-Liveness II~\ref{def:irmc-liveness2} and IRMC-Liveness III~\ref{def:irmc-liveness3}) is enough to completely replace the state of the IRMC, if necessary.
	Requests that were already contained in the IRMC must be identical as the message sent for a specific sequence number~$s$ in \texttt{ag.deliver} or \texttt{cp.stable\_cp}~(L.~\ref{def:pseudo-agreement}.\ref{code:agree-deliver-send} and \ref{def:pseudo-agreement}.\ref{code:agree-cp-send}) must be identical per induction assumption.
\end{proof}
\begin{remark}
	\texttt{$c_\mathcal{E}$.move\_window} (L.~\ref{def:pseudo-agreement}.\ref{code:agree-cp-move}) is actually called with $s-|hist'|+1$ which has the same effect as $s-|c_{\mathcal{E},0}|+1$ such that we assume $|hist'| = |c_{\mathcal{E},0}|$ in the following to simplify the presentation of the proof.
	As the first delivered agreement sequence number is 1 and for every delivered request a new message is added to $hist$~(L.~\ref{def:pseudo-agreement}.\ref{code:agree-hist}), the size of $|hist| = min(s_n,|c_{\mathcal{E},0}|)$.
	Thus when applying a checkpoint $s - |hist'|+1 = s - min(s,|c_{\mathcal{E},0}|)+1 = max(1, s-|c_{\mathcal{E},0}|+1)$.
	As $min(c_{\mathcal{E},0}.win)$ is initialized with $1$ and \texttt{$c_\mathcal{E}$.move\_window} ignores calls which move the window backwards, $s-|c_{\mathcal{E},0}|+1$ is equivalent to $s - |hist'|+1$.
\end{remark}

\subsubsection{Execution Safety (E-Safety)}
To prove property E-Safety~\ref{def:e-safety} we start with the following lemma:
\begin{lemma}
	\label{def:e-safety-commit}
	When two \execs~$e_1$ and $e_2$ receive message $m$ and $m'$ at position $p$ in the commit channel, then $m = m'$.
\end{lemma}
\begin{proof}
	\label{proof:e-safety}
	We prove this by contradiction. Assume that $m\neq m'$.
	Per IRMC-Correctness I~\ref{def:irmc-correctness1} \texttt{$c_\mathcal{E}$.receive}($0, p$)~(L.~\ref{def:pseudo-execution}.\ref{code:exec-receive}) only delivers a message $m$ that was sent by a correct \agr, the same holds for $m'$.
	Therefore \texttt{$c_\mathcal{E}$.send}($0, p, m$) and \texttt{$c_\mathcal{E}$.send}($0, p, m'$) (either at L.~\ref{def:pseudo-agreement}.\ref{code:agree-deliver-send} or \ref{def:pseudo-agreement}.\ref{code:agree-cp-send}) must have been called by a correct \agr each.
	For the \texttt{$c_\mathcal{E}$.send} call in \texttt{ag.deliver}, the \ag must have delivered message~$m$ and $m'$ on two correct replicas, which contradicts A-Safety~\ref{def:a-safety}.
	And according to CP-A-Equivalence~\ref{def:cp-a-equivalence} the \texttt{$c_\mathcal{E}$.send} when applying a checkpoint in \texttt{cp.stable\_cp} is equivalent to the previous send call in \texttt{ag.deliver}, which contradicts the assumption.
\end{proof}

\noindent{}With this we can prove E-Safety~\ref{def:e-safety}:
\begin{corollary}
	An \exec only executes requests received from the commit channel (compare L.~\ref{def:pseudo-execution}.\ref{code:exec-receive} - \ref{def:pseudo-execution}.\ref{code:exec-execute}) which according to Lemma~\ref{def:e-safety-commit} cannot receive different requests on different correct \execs.
\end{corollary}

\subsubsection{Execution Checkpoint Equivalence (CP-E-Equi-valence)}
\begin{definition}[CP-E-Equivalence]
	\label{def:cp-e-equivalence}
	The state of an \exec~($s_n$, $app$ and $u$) that has reached sequence number~$s_n$  via processing the corresponding \textsc{Execute} message~(L.~\ref{def:pseudo-execution}.\ref{code:exec-exec-msg}) for $s_n$ is equivalent to that of a replica that arrives there via a checkpoint for sequence number~$s_n$.
\end{definition}
\noindent{}The proof follows along the lines of CP-A-Equivalence~\ref{def:cp-a-equivalence}.
\begin{proof}
	We prove this by induction.\\
	\indent{}\emph{Base case}: All correct \execs initialize $s_n$, $app$ and $u$ with identical values.
	There is no checkpoint for that sequence number, as no checkpoint was generated yet.\\
	\indent{}\emph{Induction step}: All correct \execs pass through the same states or jump forward to one of those states via \linebreak a checkpoint.
	
	As updates to the considered state parts are only made in either the main loop~(L.~\ref{def:pseudo-execution}.\ref{code:exec-main}) or \texttt{cp.stable\_cp}~(L.~\ref{def:pseudo-execution}.\ref{code:exec-stable-cp}), it suffices to show that when either of them updates $s_n$ to a certain sequence number, then the resulting replica states are equivalent.
	Note that the sequence number~$s_n$ increases monotonically as the main loop only increments it~(L.~\ref{def:pseudo-execution}.\ref{code:exec-inc-steps}) and \texttt{cp.stable\_cp} only increases the value of $s_n$~(L.~\ref{def:pseudo-execution}.\ref{code:exec-cp-guard}).
	
	Assume that from a common starting point, replicas reach sequence number $s_n$ by processing the corresponding \textsc{Execute}-message~(L.~\ref{def:pseudo-execution}.\ref{code:exec-exec-msg}):
	As \texttt{$c_\mathcal{E}$.receive}($0, s_n+1$)~(L.~\ref{def:pseudo-execution}.\ref{code:exec-receive}) is called sequentially (without skipping) for each sequence number and per E-Safety~\ref{def:e-safety} all correct \execs process identical requests for each sequence number, the (atomic) modifications of $s_n$, $u[c]$ and $app$ in the main loop (L.~\ref{def:pseudo-execution}.\ref{code:exec-execute} and following) are identical across \execs.
	Either all correct \execs come to the identical decision to skip execution of request $r$~(L.~\ref{def:pseudo-execution}.\ref{code:exec-skip})based on $u[c]$, which must be identical across replicas as per induction assumption the replica states were identical which includes $u[c]$, or according to RSM-property~\ref{def:rsm} the \execs arrive at identical $u[c]$ and \texttt{app} for $s_n$ after processing~$r$.
	
	Therefore a call to \texttt{cp.gen\_cp}($s$, ($u$, \texttt{app}))~(L.~\ref{def:pseudo-execution}.\ref{code:exec-gen-cp}) for sequence number~$s$ has identical parameters on all correct \execs and thus per CP-Safety~\ref{def:cp-safety} \texttt{cp.stable\_cp}($s$, ($u'$, \texttt{app'}))~(L.~\ref{code:exec-stable-cp}) can only deliver that checkpoint.
	
	Applying a checkpoint for the current or an older sequence number~$s \leq s_n$ does not change $s_n$, $app$ and $u$~(L.~\ref{def:pseudo-execution}.\ref{code:exec-cp-guard}).
	Applying a checkpoint for a newer sequence number~$s > s_n$ atomically sets $s_n$, $app$ and $u$ to the state they had when the checkpoint was created~(L.~\ref{def:pseudo-execution}.\ref{code:exec-cp-update}).
	Later calls to \texttt{$c_\mathcal{E}$.receive} (L.~\ref{def:pseudo-execution}.\ref{code:exec-receive}) will request the next sequence number after the checkpoint.
	
	\texttt{$c_\mathcal{E}$.move\_window} (L.~\ref{def:pseudo-execution}.\ref{code:exec-cp-move}) will cause any \texttt{$c_\mathcal{E}$.receive} calls for an old sequence number to finish with a \textsc{TooOld} message and request a sequence number after the checkpoint on the next iteration.
\end{proof}

\subsubsection{Execution Safety II (E-Safety II)}
\label{def:e-safety2-proof}
\begin{lemma}
	\label{def:client-result}
	When a client accepts a reply for its request, then that reply is correct and correct \execs provide the same reply.
\end{lemma}
\begin{proof}
	A client waits for replies~(L.~\ref{def:pseudo-client}.\ref{code:client-loop}) from $f_e+1$  different replicas of its execution group with the same content~(L.~\ref{def:pseudo-client}.\ref{code:client-precheck} and \ref{def:pseudo-client}.\ref{code:client-check}), such that per failure assumption at least one of the replies is from a correct \exec.
	As shown in CP-E-Equivalence~\ref{def:cp-e-equivalence}, all correct \execs that process a request arrive at the same state and result.
	That result is either sent directly to the client~(L.~\ref{def:pseudo-execution}.\ref{code:exec-send-reply}) or retrieved from $u[c]$ on a request retry~(L.~\ref{def:pseudo-execution}.\ref{code:exec-resend-reply}).
\end{proof}

\noindent{}We can now prove E-Safety II~\ref{def:e-safety2}:
\begin{proof}
	In order to prove that \system provides linearizability, we have to show that requests issued at any point in time are always executed after all requests for which a client has accepted the reply, and that the execution follows the application's specification~\cite{herlihy90linearizability}.
	
	The latter part of the requirement was already shown in CP-E-Equivalence~\ref{def:cp-e-equivalence}, which uses the fact that requests are executed~(L.~\ref{def:pseudo-execution}.\ref{code:exec-execute}) in a total order.
	This also guarantees that at least one correct replica has processed the \textsc{Execute} message for each sequence number.
	An executed request must have been delivered by the \ag (see the proof in Section~\ref{proof:e-safety} for E-Safety~\ref{def:e-safety}).
	Assume that the \execs have executed request~$r$ which was ordered at sequence number~$s$.
	Now let the \execs execute a request~$r'$ afterwards which was ordered at a sequence number $s'$ with $s' < s$.
	However, as \execs only process requests in order, this contradicts the assumption that $r$ was already executed.
	Thus new requests are always ordered/executed at a sequence number higher than that of \linebreak previously executed requests.
	Per Lemma~\ref{def:client-result} a client cannot receive different replies from correct \execs.
	
	That is as soon as a single correct \exec sends a reply to the client, which by construction happens before that client has accepted the reply, later requests are always ordered at a higher sequence number.
\end{proof}

\begin{remark}
	The request IRMCs do not matter for E-Safety~\ref{def:e-safety} and E-Safety II~\ref{def:e-safety2}, as the \ag is safe independent of the input.
\end{remark}
\begin{remark}
	It is not necessary to store client messages in an execution checkpoint as a correct client keeps repeating incomplete requests, and as already executed requests are either part of a checkpoint or still available from the commit IRMC.
\end{remark}
\begin{remark}
	A correct \exec might not receive a request from a correct client when the other \execs already have processed it.
	This is the reason why \texttt{cp.stable\_cp} at \execs must push the window of a client's subchannel forward.
\end{remark}

\subsubsection{Execution Validity (E-Validity)}
E-Validity~\ref{def:e-validity} follows as a corollary:
\begin{corollary}
	Per Lemma \ref{def:e-safety-commit} an executed request must have been delivered by the \ag, and per A-Validity \ref{def:a-validity} only valid client requests are delivered that per cryptographic assumptions must originate from that client.
\end{corollary}

\subsubsection{\hspace{1mm}Execution Validity II (E-Validity II)}
Next, we prove E-Validity II~\ref{def:e-validity2}:
\begin{proof}
	This follows by construction of the main loop~(L.~\ref{def:pseudo-execution}.\ref{code:exec-main}):
	Requests which are not either the first request of a client or which do not have a higher counter value $t_c$ than the last one are skipped~(L.~\ref{def:pseudo-execution}.\ref{code:exec-skip}).
	After executing a request the latest counter for client $c$ is stored~(L.~\ref{def:pseudo-execution}.\ref{code:exec-store-rep}).
	As a request cannot have a counter value higher than its own counter value, it can be executed at most once.
	Per CP-E-Equivalence~\ref{def:cp-e-equivalence} $u$ and $app$ are always restored together, such that if the application state contains the effects of executing the write request, this fact is also reflected in $u$.
	And therefore the request will not be executed more than once.
\end{proof}

\subsubsection{Execution Liveness (E-Liveness)}
We now prove that a correct client will eventually receive a reply to its request(s).
Without loss of generality, we consider all requests to originate from the same client.
For this we show that each of the processing steps a request passes through will eventually make progress.
The lemmas assume implicitly that the client has either collected a stable reply (in which case the request processing is finished) or that it still waits for replies to its request and thus keeps resending its request.

\begin{lemma}
	\label{def:e-liveness-e-req}
	When a correct client sends a new request~$r$, then an \exec will pass it on to its request IRMC (unless it has already seen a newer request from that client).
\end{lemma}
\begin{proof}
	Assume that an \exec receives a, from its perspective, new request~(L.~\ref{def:pseudo-execution}.\ref{code:exec-recv-write}).
	By definition a request $r=\langle\textsc{Write},w,c,t_c\rangle$ sent by a correct client is correctly authenticated and signed~(L.~\ref{def:pseudo-client}.\ref{code:client-sign}) and therefore passes the MAC and signature checks~(L.~\ref{def:pseudo-execution}.\ref{code:exec-check-mac} and \ref{def:pseudo-execution}.\ref{code:exec-check-sig}).
	The counter value $t_c$ is $t_c>t_c'$, with $t_c'$ being the counter value of any older request, as a correct client always increments its counter value after accepting a reply~(L.~\ref{def:pseudo-client}.\ref{code:client-counter}).
	As $t[c]$ is only modified when the \exec receives a valid request from the client~(L.~\ref{def:pseudo-execution}.\ref{code:exec-counter}), it must contain either some older value $t_c'$ or the default of~$0$.
	(The client starts with $t_c=1$, whereas an \exec has $t[c]=0$.)
	Therefore $t_c>t[c]$ and the \exec calls \texttt{$r_\mathcal{E}$.send}($c$, $t_c$, unwrap($m$))~(L.~\ref{def:pseudo-execution}.\ref{code:exec-send}).
	
	In case the request is not new to the \exec, then the Lemma provides no assurances.
\end{proof}

\begin{lemma}
	\label{def:e-liveness-e-send}
	The send call by the \execs for the client's request IRMC will not block indefinitely.
\end{lemma}
\begin{proof}
	The send call only blocks if the request counter $t_c$ > max($r_{\mathcal{E},c}.win$), that is the upper bound of the client's request subchannel, according to the definition of the $send$ method.
	To arrive at a contradiction assume that the \texttt{$r_\mathcal{E}$.send} call~(L.~\ref{def:pseudo-execution}.\ref{code:exec-send}) blocks indefinitely.
	As a correct client sends its (new) request to all \execs, eventually $f_e+1$~correct \execs will per Lemma~\ref{def:e-liveness-e-req} have called \texttt{$r_\mathcal{E}$.send} and therefore also \texttt{$r_\mathcal{E}$.move\_window}($c$, $t_c$)~(L.~\ref{def:pseudo-execution}.\ref{code:exec-move-win}).
	Per IRMC-Liveness III~\ref{def:irmc-liveness3} eventually all \agrs will call \texttt{$r_\mathcal{E}$.move\_window}($c$, $t_c$).
	With IRMC-Liveness II~\ref{def:irmc-liveness2} it follows that \texttt{$r_\mathcal{E}$.send} returns, which contradicts the assumption.
\end{proof}

\begin{lemma}
	\label{def:e-liveness-a-receive}
	An \agr will eventually try to receive a new correct request~$r$ from a correct client (unless it has already seen a newer one or skipped it with a checkpoint).
\end{lemma}
\begin{proof}
	Lemma~\ref{def:e-liveness-e-send} has already shown that all ($\geq f_e+1$) correct \execs will \texttt{$r_\mathcal{E}$.send} the new client request~$r$ which per IRMC-Liveness~I~\ref{def:irmc-liveness1} can be received by a corresponding call on the \agrs unless it is no longer part of the window of the subchannel.
	According to IRMC-Correctness~I~\ref{def:irmc-correctness1} only request~$r$ can be received, as all correct \execs send this request.
	We therefore have to show that an \agr will call \texttt{$r_\mathcal{E}$.receive}($c$, $t^+[c]$) (L.~\ref{def:pseudo-agreement}.\ref{code:agree-client-receive})  for the right request counter value $t_c$.
	
	Assume that $t^+[c] < t_c$:
	As shown above in the proof of Lemma~\ref{def:e-liveness-e-send} all correct \agrs will eventually call \texttt{$r_\mathcal{E}$.move\_window}($c$, $t_c$), which according to the semantics of the $send$ method will cause it to return $\langle\textsc{TooOld}, t_c\rangle$ which is used to update $t^+[c]$~(L.~\ref{def:pseudo-agreement}.\ref{code:agree-too-old}) and request $t_c$ next.
	
	Assume that $t^+[c] > t_c$:
	We show that this case never applies.
	An \agr cannot have received a too new \textsc{TooOld} message and stored its counter value~(L.~\ref{def:pseudo-agreement}.\ref{code:agree-too-old}):
	Per IRMC-Correctness~II~\ref{def:irmc-correctness2} at least one \exec must have called \texttt{$r_\mathcal{E}$.move\_window} accordingly, which requires that a correct \exec has received a valid request with counter $t^+[c] > t_c$ from a correct client.
	This contradicts the assumption that the request is new.
	
	Incrementing $t^+[c]$ after having received a previous request~(L.~\ref{def:pseudo-agreement}.\ref{code:agree-next-req})  or processing it in \texttt{ag.deliver}~(L.~\ref{def:pseudo-agreement}.\ref{code:agree-next-deliver}) 
	would require a previous request with counter value $t_c' \geq t_c$, which contradicts the assumption.
	(A faulty client could cause some chaos here, but this is no problem as the effects are strictly limited to the client's subchannel.)
\end{proof}
\begin{remark}
	These properties effectively make the \texttt{$r_\mathcal{E}$.receive} call self-synchronizing.
\end{remark}

\begin{lemma}
	\label{def:e-liveness-a-deliver}
	The \ag will \texttt{ag.deliver} (L.~\ref{def:pseudo-agreement}.\ref{code:agree-deliver}) a new request~$r$ for sequence number~$s$ within bounded time or apply a checkpoint for a later or equal sequence number.
\end{lemma}
\begin{proof}
	After $f_e+1$~\execs complete their call to \texttt{$r_\mathcal{E}$.send}($c$, $t_c$, $r$)~(L.~\ref{def:pseudo-execution}.\ref{code:exec-send}) an \agr can receive request~$r$ and start the agreement process.
	
	Assume that the request~$r$ is not delivered within bounded time and is also not skipped via a checkpoint.
	The request of a correct client will eventually arrive at all correct ($\geq~f_e +1$) \execs.
	With Lemma~\ref{def:e-liveness-e-req} and \ref{def:e-liveness-e-send} it follows that $f_e+1$ correct \execs call \texttt{$r_\mathcal{E}$.send}.
	With IRMC-Liveness~I~\ref{def:irmc-liveness1}, IRMC-Correctness~I~\ref{def:irmc-correctness1} and Lemma~\ref{def:e-liveness-a-receive} it follows that all correct \agrs will eventually receive the request~$r$ or a $\langle\textsc{TooOld},t_c'\rangle$ message if \texttt{$r_\mathcal{E}$.move\_window}~(L.~\ref{def:pseudo-execution}.\ref{code:exec-move-win}) is called by $f_e+1$~\execs with $t_c'>t_c$.
	As a correct client does not issue a request with counter $t_c' > t_c$ before $r$ was executed, all correct \execs will eventually call \texttt{$r_\mathcal{E}$.move\_window} with exactly $t_c$, but no higher value, such that receiving \textsc{TooOld} would violate IRMC-Correctness~II~\ref{def:irmc-correctness2}. (Executing $r$ as is shown in the proof of Lemma~\ref{def:e-safety-commit} would require that it was delivered before by at least one correct \agr.)
	
	Thus, per IRMC-Liveness III~\ref{def:irmc-liveness3} all correct \agrs will eventually internally call \texttt{move\_window}($c$, $t_c$) on the request IRMC and $2f_a+1$ correct \agrs eventually receive request~$r$ as long as $r$ is not delivered.
	With A-Liveness~\ref{def:a-liveness} it follows that $f_a+1$ correct \agrs eventually deliver $r$, contradicting the assumption.
	
	Skipping the \texttt{ag.deliver} call via \texttt{ag.stable\_cp}~(L.~\ref{def:pseudo-agreement}.\ref{code:agree-stable-cp}) requires per CP-Safety~\ref{def:cp-safety} that at least one correct \agr created the checkpoint~(L.~\ref{def:pseudo-agreement}.\ref{code:agree-gen-cp}) and thus the \ag must already have delivered $r$.
\end{proof}

\begin{lemma}
	\label{def:e-liveness-e-exec}
	A request~$r$ delivered at sequence number~$s$ that is \texttt{$c_\mathcal{E}$.send} by $f_a+1$~correct \agrs will eventually either execute on $f_e+1$~correct \execs or on one correct \exec once a stable checkpoint with sequence number $s_{CP} \geq s$ was created.
\end{lemma}
\begin{proof}
	Assume that no stable checkpoint with sequence number $s_{CP}\geq s$ is applied at the \exec~(L.~\ref{def:pseudo-execution}.\ref{code:exec-stable-cp}) before processing $r$:
	IRMC-Liveness~I~\ref{def:irmc-liveness1} states that $f_e+1$~correct \execs receive some request or a $\langle\textsc{TooOld}, s'\rangle$ message~(L.~\ref{def:pseudo-execution}.\ref{code:exec-receive}) with $s' > s$ as $f_a$+$1$ \agrs sent the request~(L.~\ref{def:pseudo-agreement}.\ref{code:agree-deliver-send}).
	According to IRMC-Correctness~I~\ref{def:irmc-correctness1} the request can only be request~$r$ as per A-Correctness~\ref{def:a-safety} all correct agreement replicas send request~$r$.
	The \execs cannot receive the \textsc{TooOld} message as this would violate IRMC-Correctness II~\ref{def:irmc-correctness2}:
	
	Execution replicas can only call \texttt{$c_\mathcal{E}$.move\_window}($0, s_{CP}+1$) (L.~\ref{def:pseudo-execution}.\ref{code:exec-cp-move}) with $s_{CP} < s$ per assumption and thus $s_{CP} + 1 \leq s$, which does not allow \textsc{TooOld} to be returned.
	
	As the \ag delivers requests in sequence number order according to A-Order~\ref{def:a-order}, an \exec will also be able to receive any other previous request between $s_{CP}$ and $s$ and therefore will eventually try to receive~$s$.
	
	Agreement replicas call \texttt{$c_\mathcal{E}$.move\_window}($0, \hat{s} - |c_{\mathcal{E},0}|+1$) (L.~\ref{def:pseudo-agreement}.\ref{code:agree-cp-move}).
	To create an agreement checkpoint at $\hat{s}$~(L.~\ref{def:pseudo-agreement}.\ref{code:agree-gen-cp}), the window of the commit subchannel must have included $\hat{s}$ (as \texttt{$c_\mathcal{E}$.send}~(L.~\ref{def:pseudo-agreement}.\ref{code:agree-deliver-send}) would have blocked otherwise), that is $max(c_{\mathcal{E},0}.win) \geq \hat{s} \Leftrightarrow min(c_{\mathcal{E},0}.win) + |c_{\mathcal{E},0}| - 1 \geq \hat{s} \Leftrightarrow min(c_{\mathcal{E},0}.win) \geq \hat{s} - |c_{\mathcal{E},0}|+1$.
	Therefore, an agreement replica cannot advance the window of the commit IRMC unless an execution group triggered the window move before.
	However, as shown in the previous paragraph the latter would contradict the assumption.
	Therefore, $f_e+1$~correct \execs will eventually execute the request and possibly create a checkpoint.
	
	Assume that a stable checkpoint with sequence number $s_{CP}\geq s$ gets applied:
	Per CP-Correctness~\ref{def:cp-safety} at least one correct \exec must have created the checkpoint and thus have executed the request as per the previous part of the proof.
	Per CP-Liveness~\ref{def:cp-liveness} all other correct \execs will eventually receive and apply the checkpoint or have executed the request.
\end{proof}

\begin{lemma}
	\label{def:e-liveness-e-cp}
	A correct execution checkpoint at sequence number $s_{CP}$ for which $f_a+1$~\agrs delivered and called \texttt{$c_\mathcal{E}$.send}($0, s_{CP}$)~(L.~\ref{def:pseudo-agreement}.\ref{code:agree-deliver-send}) will eventually become stable~(L.~\ref{def:pseudo-execution}.\ref{code:exec-stable-cp}) unless it is superseded by a newer one.
\end{lemma}
\begin{proof}
	Assume that no such stable checkpoint exists and that it is not superseded by a newer one.
	Then per Lemma~\ref{def:e-liveness-e-exec} $f_e+1$~correct \execs will execute the request and thereby create their checkpoint messages~(L.~\ref{def:pseudo-execution}.\ref{code:exec-gen-cp}) which per CP-E-Equivalence~\ref{def:cp-e-equivalence} are identical and according to CP-Liveness II~\ref{def:cp-liveness2} will become stable.
\end{proof}
	
\begin{lemma}
	\label{def:e-liveness-e-win}
	If no progress occurs, then eventually the start of the subchannel window of the commit IRMC is $min(c_{\mathcal{E},0}.win)$ $= s_{CP} + 1$ with $s_{CP}$ being the latest stable execution checkpoint.
\end{lemma}
\begin{proof}
	Per CP-Liveness~\ref{def:cp-liveness} eventually all \execs will receive the latest stable execution checkpoint~(L.~\ref{def:pseudo-execution}.\ref{code:exec-stable-cp}) and call \texttt{$c_\mathcal{E}$.move\_window}($0, s_{CP}+1$)~(L.\ref{def:pseudo-execution}.\ref{code:exec-cp-move}).
	No correct \exec calls \texttt{$c_\mathcal{E}$.move\_window} for a higher sequence number as $s_{CP}$ is the number of the latest checkpoint.
	
	Agreement replicas call \texttt{$c_\mathcal{E}$.move\_window}($0, \hat{s} - |c_{\mathcal{E},0}|+1$) (L.~\ref{def:pseudo-agreement}.\ref{code:agree-cp-move}).
	To create an agreement checkpoint at $\hat{s}$, the window of the commit subchannel must have included $\hat{s}$ (as \texttt{$c_\mathcal{E}$.send}~(L.~\ref{def:pseudo-agreement}.\ref{code:agree-deliver-send}) would have blocked otherwise, preventing the checkpoint generation), that is $max(c_{\mathcal{E},0}.win) \geq \hat{s} \Leftrightarrow min(c_{\mathcal{E},0}.win) + |c_{\mathcal{E},0}| - 1 \geq \hat{s} \Leftrightarrow min(c_{\mathcal{E},0}.win) \geq \hat{s} - |c_{\mathcal{E},0}|+1$.
	Therefore an agreement replica cannot advance the window of the commit IRMC to a sequence number that is larger than that of the \execs' \texttt{$c_\mathcal{E}$.move\_window} calls.
	Thus all correct \agrs eventually arrive at $min(c_{\mathcal{E},0}.win) = s_{CP} + 1$ with $s_{CP}$ being the latest stable execution checkpoint.
\end{proof}

\begin{lemma}
	\label{def:e-liveness-a-send}
	Agreement replicas will eventually complete \texttt{$c_\mathcal{E}$.send}($s, r$)~(L.~\ref{def:pseudo-agreement}.\ref{code:agree-deliver-send}).
\end{lemma}
\begin{proof}
	\texttt{ag.deliver} blocks when $win$ is full~(L.~\ref{def:pseudo-agreement}.\ref{code:agree-win}).
	\texttt{AG-WIN} $\geq k_a$ and $win$ is always anchored directly after the sequence number of the last stable agreement checkpoint.
	Thus $win$ contains at least one sequence number for which a new agreement checkpoint will be created.
	
	Assume that \texttt{ag.deliver} blocks permanently on the window check.
	In that case, per assumption, there can be no stable agreement checkpoint with sequence number $s_{CP} \geq s$ and $s_{CP} \in win$, which would lead to progress.
	Therefore, as the client waits for $r$ to be executed, per Lemma~\ref{def:e-liveness-a-deliver} eventually $f_a+1$~\agrs also deliver all requests in $win$.
	That is $f_a+1$~correct agreement replicas create a new agreement checkpoint, which will become stable and moves $win$ forward.
	This contradicts the assumption.
	
	Assume that \texttt{$c_\mathcal{E}$.send}~(L.~\ref{def:pseudo-agreement}.\ref{code:agree-deliver-send}) blocks permanently, which requires that $s > max(c_{\mathcal{E},0}.win)$.
	Per A-Order~\ref{def:a-order} and CP-A-Equi-valence~\ref{def:cp-a-equivalence} it follows that all previous slots in the subchannel window are filled with requests.
	With Lemma~\ref{def:e-liveness-a-deliver} this applies to at least $f_a+1$~\agrs.
	As $|c_{\mathcal{E},0}|\geq k_e$  at least one position in the commit IRMC subchannel window is an execution checkpoint sequence number.
	Per Lemma~\ref{def:e-liveness-e-cp} this causes a new checkpoint to become stable, which according to Lemma~\ref{def:e-liveness-e-win} eventually moves the commit IRMC window forward and thus contradicts the assumption.
\end{proof}

\noindent{}Now we can prove that a correct client will eventually receive a reply to its request:
\begin{proof}
	Assume that the client does not get a reply.
	Then per Lemma~\ref{def:e-liveness-a-send} and \ref{def:e-liveness-e-exec} $f_e+1$~correct \execs will eventually have the reply in $u[c]$.
	As a correct client does not send a new request before having obtained a reply to the last one, $u[c]$ must eventually contain the reply.
	Per CP-E-Equivalence~\ref{def:cp-e-equivalence} the reply is identical on all correct \execs.
	At latest after the next request retry the client will receive the (identical) reply from $f_e+1$ correct \execs, and therefore accept the reply~(L.~\ref{def:pseudo-client}.\ref{code:client-check}), which contradicts the assumption.
\end{proof}

\begin{remark}
	An \agr will receive a request $r$ either via the request IRMC, the \ag or skip the request via a checkpoint.
\end{remark}


\subsubsection{Multiple Execution Groups}
We now generalize to $n_e \geq 1$ execution groups of which $z < n_e$ might be skipped if these are slow.
\begin{lemma}
	E-Liveness~\ref{def:e-liveness} also holds for multiple execution groups.
\end{lemma}
\begin{proof}
	Even though an \agr only waits for $n_e - z$ groups~(L.~\ref{def:pseudo-agreement}.\ref{code:agree-wait-send}) to complete \texttt{$c_\mathcal{E}$.send}, an execution group will only miss requests if the \agrs call \texttt{$c_\mathcal{E}$.move\_window}~(L.~\ref{def:pseudo-agreement}.\ref{code:agree-cp-move}) with a sequence number not yet received by a slow execution group.
	As shown in the proof of Lemma~\ref{def:e-liveness-e-win} an agreement replica can only create a checkpoint that would push the window of the commit IRMC forward if the execution group already has created a newer or matching checkpoint.
	Generalized to $n_e$ execution groups, the \texttt{$c_\mathcal{E}$.send}~(L.~\ref{def:pseudo-agreement}.\ref{code:agree-deliver-send}) calls for $n_e - z$ execution groups have to complete, before an agreement checkpoint can be created~(L.~\ref{def:pseudo-agreement}.\ref{code:agree-gen-cp}).
	Therefore an execution group that has fallen behind can always retrieve an up-to-date checkpoint from one of the $n_e - z$ up-to-date execution groups.
	
	As agreement replicas unconditionally move the commit IRMC window forward~(L.~\ref{def:pseudo-agreement}.\ref{code:agree-cp-move}), this will lead to at least $f_a+1$ \agrs calling \texttt{$c_\mathcal{E}$.move\_window} (per Lemma~\ref{def:e-liveness-a-send} a corresponding checkpoint will eventually exist and per CP-Liveness~\ref{def:cp-liveness} all correct \agrs will eventually receive it), which per IRMC-Liveness~I~\ref{def:irmc-liveness1} and IRMC-Liveness~III~\ref{def:irmc-liveness3} will eventually allow execution groups that fell behind to receive a \textsc{TooOld} message.
\end{proof}

\subsubsection{Consistency Guarantees}
\label{sec:consistency-guarantees}
We now revisit the consistency guarantees provided by \system.

\paragraph{Write Requests}
As previously shown in Section~\ref{def:e-safety2-proof}, \system provides linearizability for write requests.

\paragraph{Read Requests with Strong Consistency}
Read requests with strong consistency work like write requests with one exception:
Only the designated execution group receives the full request, whereas the other groups only get the client id $c$ and counter $t_c$.
This leads to the following observation:
\begin{lemma}
	With read requests, the content of checkpoints can vary between groups in regard to the reply stored in $u[c]$. That is CP-E-Equivalence~\ref{def:cp-e-equivalence} only applies for individual groups at a time. 
\end{lemma}
\begin{proof}
	Only the client's execution group will receive the read request and modify $u[c]$ accordingly after executing the request~(L.~\ref{def:pseudo-execution}.\ref{code:exec-store-rep}).
	All other execution groups store a placeholder in $u[c]$ which includes the request counter.
	Therefore, the reply parts of $u[c]$ can differ between groups.
	Note that this divergence is self-correcting in the sense that it will disappear after executing the next write request for that client.
\end{proof}
\begin{remark}
This does not prevent the checkpoint from being transferred between groups, as each group can still generate a valid proof for its checkpoint.
However, the global flow control could force a group to skip some requests, which might include group-specific read requests.
In that case an \exec has to tell the client to resubmit its request if necessary, based on the placeholder stored in $u[c]$.
\end{remark}

\paragraph{Read Requests with Weak Consistency}
\begin{lemma}
	Weakly consistent read requests provide one-copy serializability (assuming each request, which can access various parts of the application state, represents a transaction).
\end{lemma}
\begin{proof}
The reply to the client and state modifications  must be equivalent to those from an acyclic ordering of transactions, where each transaction is processed atomically~\cite{bernstein87concurrency}.

The application state is atomically modified with a totally ordered sequence of requests (see RSM~\ref{def:rsm}), which yields an acyclic order for all state-modifying requests and strongly-consistent read requests.
All weakly consistent read requests happen between these state modifications, which does not introduce cycles in the request ordering either.

As a correct client only accepts a reply sent by at least one correct replica, it will receive a conforming reply.
\end{proof}

\subsection{IRMC-RC}
\begin{figure}
	\vspace*{3mm}%
	\begin{lstlisting}
Sender replica %$r_s$%
%$rwin[r][sc]$% := %$[1, |IRMC_{sc}|]$\hfill%// Received windows, %$r \in R_R \cup {r_s}$%
%$awin[sc]$% := %$[1, |IRMC_{sc}|]$\hfill%// Active window%\hrule%
send(%\textsc{Subchannel}% %$sc$%, %\textsc{Position}% %$p$%, %\textsc{Message}% %$m$%):
	sleep until %$p \leq max(awin[sc])$%
	if %$p < min(awin[sc])$%: return %$\langle\textsc{TooOld}, min(awin[sc])\rangle$%
	else: // %$p \in awin[sc]$%
		send %$sign_{r_s}(\langle\textsc{Send}, m, sc, p\rangle)$% to %$R_R$%

move_window(%\textsc{Subchannel}% %$sc$%, %\textsc{Position}% %$p$%):
	// The subchannel window start may only increase
	if %$p > min(rwin[r_s][sc])$%:
		send %$mac_{r_s,R_R}(\langle\textsc{Move}, sc, p\rangle)$% to %$R_R$%
		%$rwin[r_s][sc]$% := %$[p, p+|IRMC_{sc}|-1]$%

on receive(%$m$% = %$\langle\textsc{Move}, sc, p\rangle$% from %$r_r \in R_R$%):
	if %$!valid\_mac_{r_r,R_S}$%(%$m$%): return
	// Only accept new move messages
	if %$p > min(rwin[r_r][sc])$%:
		%$rwin[r_r][sc]$% := %$[p, p + |IRMC_{sc}|-1]$%
		// Calculate actual window start
		w := %$f_r + 1$% highest %$\{min(rwin[r_r'][sc])\,|\, r_r' \in R_R\}$%
		%$awin[sc]$% := %$[w, w+|IRMC_{sc}|-1]$%
		garbage-collect messages with SeqNr %$s < awin[sc]$%
%\hrule%
Receiver replica %$r_r$%
%$rwin[r][sc]$% = %$[1, |IRMC_{sc}|]$\hfill%// Received windows, %$r \in R_S \cup {r_r}$%
%$awin[sc]$% = %$[1, |IRMC_{sc}|]$\hfill%// Active window
%$d[sc][p][r_s]$% = %$\varnothing$\hfill%// Received %\textsc{Send}% messages%\hrule%
receive(%\textsc{Subchannel}% %$sc$%, %\textsc{Position}% %$p$%) -> %\textsc{Message}% %$m$%:
	sleep until %$p \leq max(awin[sc])$%
	sleep until either:
	- case %$p < min(awin[sc])$%:
		 return %$\langle\textsc{TooOld}, min(awin[sc])\rangle$%
	- case %$\exists m: |\{r_s|r_s \in R_S, m \in d[sc][p][r_s]\}| \geq f_s + 1$%:
		return %$m$% // Received %$m$% from at least %$f_s+1$% senders

move_window(%\textsc{Subchannel}% %$sc$%, %\textsc{Position}% %$p$%):
	// The subchannel window start may only increase
	if %$p > min(awin[sc])$%:
		send %$mac_{r_r,R_S}(\langle\textsc{Move}, sc, p\rangle)$% to %$R_S$%
		%$awin[sc]$% := %$[p, p+|IRMC_{sc}|-1]$%
		garbage-collect messages with SeqNr %$s < awin[sc]$%

on receive(%$r$% = %$\langle\textsc{Send}, m, sc, p\rangle$% from %$r_s \in R_S$%):
	if %$!valid\_sig_{R_S}$%(%$r$%): return
	if %$p \geq min(awin[sc])$%:
		%$d[sc][p][r_s]$% := %$m$%

on receive(%$m$% = %$\langle\textsc{Move}, sc, p\rangle$% from %$r_s \in R_S$%):
	if %$!valid\_mac_{r_s,R_R}$%(%$m$%): return
	// Only accept new move messages
	if %$p > min(rwin[r_s][sc])$%:
		%$rwin[r_s][sc]$% := %$[p, p + |IRMC_{sc}|-1]$%
		%$nw$% := %$f_s+1$% highest %$\{min(rwin[r_s'][sc])\,|\,r_s' \in R_S\}$%
		if %$min(awin[sc]) < nw$%:
			move_window(%$s, nw$%)
	\end{lstlisting}
	\caption{IRMC-RC (pseudo code)}
	\label{def:pseudo-irmc-rc}
\end{figure}
The IRMC-RC variant shown in Figure~\ref{def:pseudo-irmc-rc} is a simple implementation of an IRMC that provides the expected properties.
Replicas can aggregate \textsc{Move} messages before sending them.
In case a sender replica has multiple IRMCs and sends identical messages on the same subchannel and position, then it can share a single signed \textsc{Send} message between IRMCs.

Without loss of generality we assume the set of senders~$R_S$ and receivers~$R_R$ to be disjoint, that is $R_S \cap R_R = \varnothing$.
We assume reliable point-to-point channels between replicas, that is messages sent between individual replicas will be delivered eventually, unless messages are garbage collected at which point a replica discards old messages, even when they were not successfully delivered yet.
To keep the pseudo code short, we assume that messages without correct authentication are automatically dropped before these can be processed.

All messages are also expected to contain an identifier to allow differentiation between different IRMCs if necessary.

\subsection{IRMC-SC}
\begin{figure*}[!htp]
	\vspace*{3mm}%
\begin{minipage}[t]{1.0\columnwidth}
	\begin{lstlisting}
Sender replica %$r_s$%
+ Variables %from% IRMC-RC
%$sig[sc][p][r_s]$% = %$\varnothing$\hfill% // Certificate share from sender %$r_s$% for%\\% %\hfill%subchannel %$sc$% position %$p$%
%$bundle[sc][p]$%  = %$\varnothing$\hfill%// %$\textsc{Certificate}$% for subchannel %$sc$% position %$p$%
%$sender[sc][r_r]$% = %$\bot$\hfill%// Selected%\,%sender%\,%for%\,%subchannel%\,$sc$\,%to%\,%receiver%\,$r_r$%
%$d[sc][p]$% = %$\varnothing$%%\hfill%// Message sent in subchannel %$sc$% at position %$p$%%\hrule%
send(%\textsc{Subchannel}% %$sc$%, %\textsc{Position}% %$p$%, %\textsc{Message}% %$m$%):
	sleep until %$p \leq max(awin[sc])$%
	if %$p < min(awin[sc])$%: return %$\langle\textsc{TooOld}, min(awin[sc])\rangle$%
	else: // %$p \in awin[sc]$%
		%$d[sc][p]$% := %$m$%
		// %$\textsc{SigShare}$% is also processed locally
		send %$sign_{r_s}(\langle\textsc{SigShare}, h(m), sc, p\rangle)$% to %$R_S$%

on receive(%$sg$% = %$\langle\textsc{SigShare}, h(m), sc, p\rangle$% from %$r_s \in R_S$%):
	if %$!valid\_sig_{R_S}$%(%$sg$%): return
	if %$p \geq min(awin[sc]) \wedge sig[sc][p][r_s]$% = %$\varnothing$%: // Only accept first share per sender
		%$sig[sc][p][r_s]$% := %$sg$%
		%$v$% := %$\{sig[sc][p][r]\,|\,r \in R_S, sig[sc][p][r].h = h(m)\}$%
		limit %$v$% to %$f_s+1$% values
		// Check if replica has %$f_s\hspace{-.3mm}+\hspace{-.3mm}1$% matching shares and the actual request
		if %$|v| = f_s + 1 \wedge d[sc][p] \neq \varnothing \wedge bundle[sc][p] = \varnothing$%:
			%$bundle[sc][p]$% := %$mac_{r_s,R_R}(\langle\textsc{Certificate}, d[sc][p], sc, p, v\rangle)$%
			send %$bundle[sc][p]$% to receivers  $r$ where %$sender[sc][r]$% = %$r_s$%

periodic:
	// Send position of latest certificate per subchannel with no gaps at previous positions in the subchannel window
	for each subchannel %$sc$%:
		 %$prog[sc]$% := highest %$p \in awin[sc]$% with %$\forall p' \in awin[sc], p'\leq p: bundle[sc][p'] \neq \varnothing$%
	send %$mac_{r_s,R_R}(\langle\textsc{Progress},prog\rangle)$% to %$R_R$%

// move_window and receive(%\textsc{Move}%) are identical to IRMC-RC

// Select sender for subchannel
on receive(%$m$% = %$\langle\textsc{Select}, sc, s\rangle$% from %$r_r \in R_R$%):
	if %$!valid\_mac_{r_r,R_S}$%(%$m$%): return
	%$sender[sc][r_r]$% := %$s$%
	// Send queued messages for subchannel %$sc$% to %$r_r$%
	%$\forall  p:$% send %$bundle[sc][p]$% to receiver  $r_r$ if %$s = r_s$%
	\end{lstlisting}
	\caption{IRMC-SC sender endpoint (pseudo code)}
	\label{def:pseudo-irmc-sc-sender}
\end{minipage}
\hfill%
\begin{minipage}[t]{1.0\columnwidth}
\begin{lstlisting}
Receiver replica %$r_r$%
+ Variables %from% IRMC-RC
%$d[sc][p]$% = %$\varnothing$%%\hfill%// Message received for subchannel %$sc$% at position %$p$%
%$pe[r][sc]$% := 0 %\hfill%// Individual expected progress reported by %$r \in R_S$%
%$pm[sc]$% := 0 %\hfill%// Merged progress values (%$f_s+1$% highest)%\hrule%
receive(%\textsc{Subchannel}% %$sc$%, %\textsc{Position}% %$p$%) -> %\textsc{Message}% %$m$%:
	sleep until %$p \leq max(awin[sc])$%
	sleep until either:
	- case %$p < min(awin[sc])$%:
		return %$\langle\textsc{TooOld}, min(awin[sc])\rangle$%
	- case %$d[sc][p] \neq \varnothing$%:
		return %$d[sc][p]$%

on receive(%$r$% = %$\langle\textsc{Certificate}, m, sc, p, v\rangle$% from %$r_s \in R_S$%):
	if %$!valid\_mac_{r_s,R_R}$%(%$r$%): return
	// Certificate must contain %$f_s+1$% matching signatures from different sender endpoints
	if %$p \geq min(awin[sc]) \wedge |v| = f_s + 1 \wedge \forall sg \in v:\linebreak{}\hspace{-4.5mm}valid\_sig_{R_S}(sg\textrm{ for }m) \wedge sg$% %from% different senders:
		%$d[sc][p]$% := %$m$%

on receive(%$m$% = %$\langle\textsc{Progress}, np\rangle$% from %$r_s \in R_S$%):
	if %$!valid\_mac_{r_s,R_R}$%(%$m$%): return
	// Merge progress vectors
	for each subchannel %$sc$%:
		%$pe[r_s][sc]$% := %$max(pe[r_s][sc], np[sc])$%
		%$pm[sc]$% := %$f_s+1$% highest %$\{pe[r'][sc]\,|\,r' \in R_S\}$%
	// Start timeout if some messages are still missing
	if %$\exists s' \in [min(awin[sc]), pm[sc]]: d[sc][s'] = \varnothing$%:
		start timer %for% %$sc@pm[sc]$%, if not started yet

on timeout for %$sc@pm[sc]$%:
	// Timeout expired and there are still missing certificates
	if %$\exists s' \in [min(awin[sc]), pm[sc]]: d[sc][s'] = \varnothing$%:
		select new sender %$r_s$% for %$sc$%
		send %$mac_{r_r,R_S}(\langle\textsc{Select},sc, r_s\rangle)$% to %$R_S$%
		restart timer %for% %$sc@pm[sc]$%

// move_window and receive(%\textsc{Move}%) are identical to IRMC-RC
	\end{lstlisting}
	\caption{IRMC-SC receiver endpoint (pseudo code)}
	\label{def:pseudo-irmc-sc-receiver}
	\end{minipage}
\end{figure*}

IRMC-SC shown in Figure~\ref{def:pseudo-irmc-sc-sender} and \ref{def:pseudo-irmc-sc-receiver} is a more complex but also more efficient implementation than IRMC-RC.

For liveness, we assume that the \textsc{Move}~message is protected against replay attacks, for example by including a counter to filter out already processed instances of the message to ensure that these are not processed multiple times.
In case a sender replica has multiple IRMCs and sends identical messages on the same subchannel and position, then it can share a single signed \textsc{Certificate} message between IRMCs.

\vfill
\pagebreak
\vfill